\documentclass[11pt]{amsart}

\usepackage[margin=1in]{geometry}
\usepackage{amssymb,amsmath,amsthm}
\usepackage{enumitem,bm,graphicx,hyperref}
\usepackage{autonum}
\usepackage{amsrefs}
\newcommand{\R}{\mathbb{R}}
\newcommand{\C}{\mathbb{C}}
\newcommand{\Z}{\mathbb{Z}}

\newcommand{\cI}{\mathcal{I}}
\newcommand{\cJ}{\mathcal{J}}

\newcommand{\bra}[1]{\la #1|}
\newcommand{\ket}[1]{| #1 \ra}

\renewcommand{\vec}[1]{\boldsymbol{#1}}

\newcommand{\la}{\langle}
\newcommand{\ra}{\rangle}
\DeclareMathOperator{\range}{range}
\DeclareMathOperator{\tr}{tr}

\DeclareMathOperator{\dist}{dist}

\DeclareMathOperator{\diam}{diam}

\newcommand{\dd}{\,\ensuremath{\textrm{d}}}

\newcommand{\comment}[1]{}

\usepackage{xcolor}
\definecolor{purp}{RGB}{160, 32, 240}
\definecolor{lightblue}{RGB}{32, 160, 240}

\newtheorem{assumption}{Assumption}
\newtheorem{theorem}{Theorem}
\newtheorem{conjecture}{Conjecture}
\newtheorem{lemma}{Lemma}[section]
\newtheorem*{lemma*}{Lemma}
\newtheorem{proposition}[lemma]{Proposition}
\newtheorem*{proposition*}{Proposition}

\theoremstyle{definition}

\newtheorem{definition}{Definition}

\theoremstyle{theorem}

\allowdisplaybreaks
\title{Algebraic localization implies exponential localization in non-periodic insulators}

\date{February 1, 2022}
\thanks{This work is supported in part by the U.S.~National Science Foundation via grant DMS-2012286 and the U.S.~Department of Energy via grant DE-SC0019449. K.D.S. was supported in part by a National Science Foundation Graduate Research Fellowship under Grant No. DGE-1644868. We would like to thank Alexander B. Watson for helpful discussions. 
}

\author{Jianfeng Lu}
\address{(JL) Department of Mathematics, Department of Physics, and Department of Chemistry, Duke University, Box 90320, Durham, NC 27708, USA}
\email{jianfeng@math.duke.edu}

\author{Kevin D. Stubbs}
\address{(KDS) Department of Mathematics, Duke University, Box 90320, Durham, NC 27708, USA}
\email{kstubbs@math.duke.edu}

\begin{document}

\begin{abstract}
  Exponentially-localized Wannier functions are a basis of the Fermi
  projection of a Hamiltonian consisting of functions which decay
  exponentially fast in space. In two and three spatial dimensions, it
  is well understood for periodic insulators that
  exponentially-localized Wannier functions exist if and only if there
  exists an orthonormal basis for the Fermi projection with finite
  second moment (i.e. all basis elements satisfy
  $\int |\vec{x}|^2 |w(\vec{x})|^2 \dd{\vec{x}} < \infty$). In this
  work, we establish a similar result for non-periodic insulators in
  two spatial dimensions. In particular, we prove that if there exists
  an orthonormal basis for the Fermi projection which satisfies
  $\int |\vec{x}|^{5 + \epsilon} |w(\vec{x})|^2 \dd{\vec{x}} < \infty$
  for some $\epsilon > 0$ then there also exists an orthonormal basis
  for the Fermi projection which decays exponentially fast in
  space. This result lends support to the Localization Dichotomy
  Conjecture for non-periodic systems recently proposed by Marcelli,
  Monaco, Moscolari, and Panati in
  \cites{2019MarcelliMonacoMoscolariPanati,2020MarcelliMoscolariPanati}.
\end{abstract}

\maketitle

\section{Introduction}
\label{sec:introduction}
In electronic structure theory, we are often interested in localized representations of the occupied space. For band-insulating materials (i.e., materials where the Fermi energy lies the in the band gap) these local representations can be studied by constructing a well localized basis for a subspace known as the Fermi projection. Over the past few decades, exponentially localized Wannier functions (ELWFs) have emerged as the basis of choice for the Fermi projection in periodic insulators. Among the many applications, Wannier functions have been instrumental in constructing effective tight binding models and in the development of the modern theory of polarization. We refer the readers to the review article  \cites{2012MarzariMostofiYatesSouzaVanderbilt} and references therein for more detailed discussions on the Wannier functions from a physical point of view. 

Given the importance of ELWFs, there have been much research devoted to understanding when ELWFs exist. As a result of these efforts, it is now well understood that for periodic insulators in dimension one, two, and three the existence of ELWFs is tied to the vanishing of certain topological invariants. We give a quick review of these theoretical results 
\begin{itemize}
\item In one spatial dimension, a basis of ELWFs for Fermi projection of an insulating crystalline material always exists \cites{1964DesCloizeaux,1964DesCloizeaux2,1983Nenciu,1988HelfferSjostrand,1991Nenciu}. 
\item In two dimensions, the same result holds if and only if the Chern number, a topological invariant associated with the Fermi projector, vanishes \cites{2007BrouderPanatiCalandraMourougane,2007Panati,2018MonacoPanatiPisanteTeufel}. 
\item In three dimensions, the result holds as long as three ``Chern-like'' numbers all vanish \cites{2007BrouderPanatiCalandraMourougane,2007Panati,2018MonacoPanatiPisanteTeufel}. 
\end{itemize}

To complement these results which connect ELWFs and topology, in \cites{2018MonacoPanatiPisanteTeufel} the authors were able to show that for periodic systems in two dimensions that the Chern number vanishes if and only if there exists a basis of Wannier functions with finite second moment (and similarly in three dimensions). This result, combined with the previous results connecting ELWFs to the Chern number, forms the basis of the localization-topology correspondence or \textit{Localization Dichotomy}. Informally stated, this correspondence is the following result in two and three dimensions:
\begin{align}
    \exists &\text{ an orthonormal basis, $\{ w_\alpha \}$, for $\range{(P)}$ s.t. for all $\alpha$,} \int_{\R^2} |\vec{x}|^2 |w_\alpha(\vec{x})|^2 \dd{\vec{x}} < \infty \\ 
    & \Longleftrightarrow \text{The Chern number vanishes / The ``Chern-like'' numbers vanish} \\
    & \Longleftrightarrow \exists\, \gamma^* > 0, \exists \text{ an orthonormal basis, $\{ \tilde{w}_\alpha \}$, for $\range{(P)}$ s.t.,}  \\
    & \hspace{15em} \text{for all $\alpha$,} \qquad \int_{\R^2} e^{2\gamma^* |\vec{x}|} |\tilde{w}_\alpha(\vec{x})|^2 \dd{\vec{x}} < \infty
\end{align}

Far less is known for systems which are not periodic however it has been conjectured that a Localization Dichotomy should also hold in non-periodic systems \cites{2019MarcelliMonacoMoscolariPanati}. The main result of the paper is to establish a weaker Localization Dichotomy for both periodic and non-periodic systems. In particular, we show that that one can establish the equivalence between exponential and algebraic localization without connecting to the theory of topological invariants. Using the informal notation introduced above, the main result of this paper is to show in two dimensions that:
\begin{align}
    \exists &\text{ an orthonormal basis, $\{ w_\alpha \}$, for $\range{(P)}$ s.t. for some $\epsilon > 0$ and} \\
    & \hspace{15em} \text{for all $\alpha$,}\qquad \int_{\R^2} |\vec{x}|^{5 + \epsilon} |w_\alpha(\vec{x})|^{2} \dd{\vec{x}} < \infty, \\ 
    & \Longleftrightarrow \exists\, \gamma^* > 0, \exists \text{ an orthonormal basis, $\{ \tilde{w}_\alpha \}$, for $\range{(P)}$ s.t.,} \\
    & \hspace{15em} \text{for all $\alpha$,} \qquad \int_{\R^2} e^{2\gamma^* |\vec{x}|} |\tilde{w}_\alpha(\vec{x})|^2 \dd{\vec{x}} < \infty
\end{align}
Unlike many previous results on the existence of Wannier functions, our main result makes no assumptions about the underlying symmetries of the system. Instead, we directly study the localization properties of the Fermi projector $P$, by constructing a new self-adjoint position operator $\widehat{X}$ which depends on $P$. Our constructed position operator is close to the true position operator $X$ in the standard coordinate basis, acts locally, and can be used to localize the $X$ position of the states in the range of the Fermi projector. As we will show, the fact that such an operator $\widehat{X}$ can be constructed implies that exponentially localized Wannier functions exist in both periodic and non-periodic systems and is equivalent to the system being topologically trivial for periodic systems.

This result lends support to the Localization Dichotomy Conjecture for non-periodic systems recently proposed by Marcelli, Monaco, Moscolari, and Panati in \cites{2019MarcelliMonacoMoscolariPanati,2020MarcelliMoscolariPanati}. The precise statement of these results follows in the next section.

\subsection{Technical Statement of Results}
Our main theoretical result concern operators which admit exponentially localized kernels in the following:
\begin{definition}
\label{def:exp-loc-kern}
  Suppose that $A$ is a bounded linear operator on $L^2(\R^d) \rightarrow L^2(\R^d)$. We say that $A$ admits an \textit{exponentially localized kernel} with rate $\gamma$, if $A$ admits an integral kernel $A(\cdot, \cdot) : \R^d \times \R^d \rightarrow \C$ and there exists a finite, positive constant $C$ so that:
\[
    |A(\vec{x}, \vec{x}')| \leq C e^{-\gamma | \vec{x} - \vec{x}'|} \quad a.e.
\]
\end{definition}
In particular, we will consider orthogonal projectors $P$ which admit exponentially localized kernels. Such projectors naturally arise the study of single-body magnetic Schr{\"o}dinger operators. As a concrete example of this, we state the following proposition : 
\begin{proposition}[\cites{1982SimonSchrodinger,Moscolari2019,2020MarcelliMoscolariPanati}]
\label{prop:h}
Let $H$ be a Hamiltonian densely defined on $L^2(\R^2)$ of the following form:
  \[
    H = (i \nabla + A)^2 + V.
  \]
  Suppose that $A \in L^4_{loc}(\R^2; \R^2)$, $\text{div}(A) \in L^2_{\text{loc}}(\R^2;\R)$, and $V \in L^2_{\text{uloc}}(\R^2;\R)$ where for $d \in \Z^+$
  \begin{align}
  & L^p_{loc}(\R^2; \R^d) := \{ f : \forall \vec{x} \in \R^2,\, \int_{|\vec{x} - \vec{y}| \leq 1} |f(\vec{y})|^p \dd{\vec{y}} < \infty \} \\[1ex] 
  & L^p_{uloc}(\R^2; \R^d) := \{ f : \sup_{\vec{x} \in \R^2} \int_{|\vec{x} - \vec{y}| \leq 1} |f(\vec{y})|^p \dd{\vec{y}} < \infty \}.
  \end{align}
  That is, $L^p_{loc}$ denotes the set of functions which are $L^p$ integrable over any compact set and and $L^p_{uloc}$ denotes the set of functions whose $L^p$ integral over any unit ball is uniformly bounded by a constant. Then $H$ is essentially self-adjoint on $L^2(\R^2)$
  
  Suppose further that we can decompose the spectrum of $H$ as $\sigma(H) = \sigma_0 \cup \sigma_1$ where $\dist{(\sigma_0, \sigma_1)} > 0$ and $\diam{(\sigma_0)} < \infty$. If $P$ is the spectral projector associated with the set $\sigma_0$ then $P$ admits an exponentially localized kernel in the sense of Definition~\ref{def:exp-loc-kern}.
\end{proposition}

With the example of gapped Schr{\"o}dinger operators in mind, we now state our main result.
As part of these steps, we define the Japanese bracket $\la \vec{x} \ra := \sqrt{1 + | \vec{x} |^2}$.

\begin{definition}[$s$-localized generalized Wannier basis]
\label{def:s-loc}
Given an orthogonal projector $P$, we say an orthonormal basis $\{ \psi_{\alpha} \}_{\alpha \in \cI} \subseteq L^2(\R^2)$ is an \textit{$s$-localized generalized Wannier basis for $P$} for some $s > 0$ if:
\begin{enumerate}
\item The collection $\{ \psi_{\alpha} \}_{\alpha \in \cI}$ spans $\range{(P)}$,
\item There exists a finite, positive constant $C$ and a collection of points $\{ \vec{\mu}_{\alpha} \}_{\alpha \in \cI} \subseteq \R^2$ such that for all $\alpha \in \cI$
\[
\int_{\R^2} \la \vec{x} - \vec{\mu}_{\alpha} \ra^{2s} |\psi_{\alpha}(\vec{x})|^2 \dd{\vec{x}} \leq C.
\]
We refer to the collection $\{ \vec{\mu}_{\alpha} \}$ as the \textit{center points} for $\{ \psi_{\alpha} \}$
\end{enumerate}
\end{definition}

\begin{definition}[Exponentially localized generalized Wannier basis]
\label{def:exp-loc}
Given an orthogonal projector $P$, we say an orthonormal basis $\{ \psi_{\alpha} \}_{\alpha \in \cI} \subseteq L^2(\R^2)$ is an \textit{exponentially localized generalized Wannier basis for $P$} if:
\begin{enumerate}
\item The collection $\{ \psi_{\alpha} \}_{\alpha \in \cI}$ spans $\range{(P)}$,
\item There exists finite, positive constants $(C, \gamma)$ and a collection of points $\{ \vec{\mu}_{\alpha} \}_{\alpha \in \cI} \subseteq \R^2$ such that for all $\alpha \in \cI$
\[
\int_{\R^2} e^{2\gamma \la \vec{x} - \vec{\mu}_{\alpha} \ra} |\psi_{\alpha}(\vec{x})|^2 \dd{\vec{x}} \leq C.
\]
As before, we refer to the collection $\{ \vec{\mu}_{\alpha} \}$ as the \textit{center points} for $\{ \psi_{\alpha} \}$
\end{enumerate}
\end{definition}

With these definitions, the main result of this paper is the following:
\begin{theorem}[Main Theorem]
  \label{thm:main}
Let $P$ be any orthogonal projector on $L^2(\R^2)$ which admits an exponentially localized kernel (Definition~\ref{def:exp-loc-kern}). Then the following statements are equivalent:
  \begin{enumerate}
  \item $P$ admits an exponentially localized generalized Wannier basis.
  \item $P$ admits an $s$-localized generalized Wannier basis for $s > 5/2$.
  \end{enumerate}
\end{theorem}

\subsection{Connection with Previous Work and Discussions}
\label{sec:conn-with-prev-work}
As discussed in the beginning of this section, for a two dimensional gapped, periodic Sch{\"o}dinger operator, it is now well understood that a collection of isolated bands admit a well localized basis if and only if the Chern number associated with those bands vanish. Given the nature of topological invariants, it is natural to expect that the localization dichotomy should also be true under small perturbations which do not close the gap of the Hamiltonian. The main result of this paper shows a weaker form of the localization dichotomy for \textit{all} gapped systems, not just those which are close to periodic.

Since we make no symmetry assumptions on our underlying system, the usual method of applying the Bloch-Floquet transform to construct Wannier functions cannot be applied. Instead, we use an alternate technique for constructing (generalized) Wannier functions first proposed by Kivelson and was later expanded by Nenciu-Nenciu\footnote{We will only highlight the key ideas of the  Kivelson and Nenciu-Nenciu return to discuss these works in more detail in Section \ref{sec:proj-pos}.}. In these works, the authors proposed defining the Wannier functions of a one dimensional system to be the eigenfunctions of the operator $PXP$. For periodic systems, this definition agrees with the usual definition of Wannier functions and are exponentially localized \cite{1998NenciuNenciu}.

While the eigenfunctions of $PXP$ are a simple way of constructing ELWFs in one dimension, extending this idea to two dimensions is not straightforward due to the existence of topological obstructions. Following our previous work \cite{2020StubbsWatsonLu1}, we construct ELWFs by assuming it is possible to decompose the projector $P$ into a family of mutually orthogonal projectors $\{ P_j \}_{j \in \cJ}$ so that each $P_j$ (1) admits an exponentially localized kernel, and (2) is localized in $X$. Due to the localization in $X$, the range of each $P_j$ is effectively one dimensional. Therefore, by extending the results of Kivelson and Nenciu-Nenciu, it can be shown that for each $j \in \cJ$ eigenfunctions of $P_j Y P_j$ are exponentially localized in both $X$ and $Y$. Hence, to construct ELWFs in two dimensions it is sufficient to construct the family $\{ P_j \}_{j \in \cJ}$.

In periodic systems, the ability to decompose $P$ into the family $\{ P_j \}_{j \in \cJ}$ is both necessary and sufficient for the system to be topologically trivial as a consequence of the localization dichotomy. To see why this is the case, suppose that the set $\{ \psi^{(1)}, \psi^{(2)}, \cdots, \psi^{(N)} \}$ generates a basis of Wannier functions in an $N$-band system under the group of translations $\{ T_{(j,k)} :  (j,k) \in \Z^2 \}$ where $T_{(j,k)}$ inherits the group structure of $\Z^2$ in the obvious way. In this case, the projectors $\{ P_j \}_{j \in \Z}$ can be constructed by the formula:
\[
  P_j := \sum_{k \in \Z} \sum_{n=1}^N T_{(j,k)} \ket{\psi^{(n)}} \bra{\psi^{(n)}}T_{(j,k)}^\dagger.
\]
If we additionally assume that $\psi^{(n)}$ decays exponentially quickly in space, then one easily checks that $P_j$ admits an exponentially localized kernel and is localized in $X$. The key point of our work is that the family of projectors $\{ P_j \}$ can be constructed in both periodic and non-periodic systems \textit{without} assuming the existence of an exponentially localized basis; fast enough algebraic decay suffices. 

Generalized Wannier bases in two dimensions have been considered and defined in a few previous works \cites{1998NenciuNenciu,2019MarcelliMonacoMoscolariPanati,2020MarcelliMoscolariPanati}. For systems without additional symmetries, the definition we give in this work agrees with the definition given by Nenciu-Nenciu in \cites{1998NenciuNenciu}. In contrast, the recent works by Marcelli, Monaco, Moscolari, and Panati \cites{2019MarcelliMonacoMoscolariPanati} and Marcelli, Moscolari and Panati \cites{2020MarcelliMoscolariPanati} have both assumed that the center points of the generalized Wannier basis lie on a set with fairly rigid structure. In our recent work, we showed that $P$ admitting an exponentially localized kernel implies that, from the perspective of localization, we can assume without loss of generality that the center points of the generalized Wannier basis lie on the integer lattice \cite{2021LuStubbs}. We will review the proof of this reduction in Section \ref{sec:bdd-density}.

To conclude this section, we recall Localization Dichotomy Conjecture for Non-Periodic systems introduced in \cites{2019MarcelliMonacoMoscolariPanati}. The natural extension of the Chern number to the non-periodic case was introduced in \cites{2019MarcelliMonacoMoscolariPanati} (see also \cites{2018CorneanMonacoMoscolari}) and is known as \textit{Chern marker}.
\begin{definition}[Chern Marker]
Let $P$ be a projection on $L^2(\R^2)$ and $\chi_L$ be the indicator function of the set $(-L,L]^2$.  The \textit{Chern marker} of $P$ is defined by
  \[
    C(P) := \lim_{L \rightarrow \infty} \frac{2 \pi i}{4 L^2} \tr{\left( \chi_{L} P \Big[ [ X, P ], [ Y, P ] \Big] P \chi_{L}\right)}
  \]
  whenever the limit on the right hand side exists.
\end{definition}
With this definition, the Localization Dichotomy Conjecture as stated in \cites{2019MarcelliMonacoMoscolariPanati} is the following:
\begin{conjecture}[Localization Dichotomy Conjecture]
  Let $P$ be the spectral projector onto $\sigma_0$ for a Hamiltonian $H$ satisfying the assumptions of Proposition \ref{prop:h}. Then the following statements are equivalent:
  \begin{enumerate}[label=(\alph*)]
  \item $P$ admits a generalized Wannier basis that is exponentially localized.
  \item $P$ admits a generalized Wannier basis that is $s$-localized for $s = 1$.
  \item $P$ is topologically trivial in the sense that its Chern marker $C(P)$ exists and is equal to zero.
  \end{enumerate}
\end{conjecture}
The main result of this work shows that (b) $\Rightarrow$ (a) for $s > 5/2$. In recent work \cites{2020MarcelliMoscolariPanati}, Marcelli, Moscolari, and Panati have shown that (b) $\Rightarrow$ (c) for $s > 5$ which was later improved to $s > 1$ in \cite{2021LuStubbs}. We remark that it is likely that the condition $s > 5/2$ can be weakened for (b) $\Rightarrow$ (a), while new technical ingredient is needed beyond the current analysis. 

While in this work we only consider systems in two dimensions, we expect that the techniques we use here can be generalized to arbitrarily high dimension by an inductive procedure (see \cite[Section 8]{2020StubbsWatsonLu1}). In particular, we predict using these techniques one can establish in $d$-dimensions that (b) $\Rightarrow$ (c) for $s > d + \frac{1}{2}$. We leave this intriguing direction to future research.

\section{Notation and Conventions}
\label{sec:notation}
We begin by fixing some notations. Vectors in $\R^d$ will be denoted by bold face with their components denoted by subscripts. For example, $\vec{v} = (v_1, v_2, v_3, \cdots, v_d) \in \R^d$. For any $\vec{v} \in \R^d$, we use $| \cdot |$ to denote its Euclidean norm. That is $| \vec{v} | := \bigl(\sum_{i=1}^d v_i^2\bigr)^{1/2}$.
For any $f : \R^2 \rightarrow \C$, we will use $\|f\|$ to denote the $L^2$-norm. For any linear operator $A$ on $L^2(\R^2)$, we will use $\|A\|$ to denote the spectral norm. 

Given a point $\vec{a} \in \R^2$, and a non-negative constant $\gamma \geq 0$, we define an exponential growth operator, $B_{\gamma,\vec{a}}$, by 
\begin{equation}
  B_{\gamma,\vec{a}} := \exp\left( \gamma \sqrt{ 1 + (X-a_1)^2 + (Y-a_2)^2 } \right) = e^{\gamma \lvert (X - a_1,\, Y - a_2,\, 1) \rvert},
\end{equation}
where $X$ and $Y$ are the standard position operators: $X f(\vec{x}) = x_1 f(\vec{x})$, $Y f(\vec{x}) = x_2 f(\vec{x})$.
Given a linear operator $A$, we define 
\begin{equation} \label{eq:notation_1}
  A_{\gamma,\vec{a}} := B_{\gamma,\vec{a}} A B_{\gamma,\vec{a}}^{-1}.
\end{equation}
We refer to $A_{\gamma,\vec{a}}$ as ``exponentially-tilted'' relative to $A$. We will often prove estimates where we use the notation~\eqref{eq:notation_1} but omit the point $\vec{a}$. In this case the estimate should be understood as uniform in the choice of point $\vec{a}$. As a note, per our convention, when $\gamma = 0$, $A_{\gamma,\vec{a}} = A$. 

Throughout this paper, we will assume that the orthogonal projector $P$ is fixed and admits an exponentially localized kernel. For this projection, we also define the operators $Q$, $P_{\gamma,\vec{a}}$, $Q_{\gamma,\vec{a}}$:
\begin{align}
    & Q := I - P \\
    & P_{\gamma,\vec{a}} := B_{\gamma,\vec{a}} P B_{\gamma,\vec{a}}^{-1} \\
    & Q_{\gamma,\vec{a}} := B_{\gamma,\vec{a}} Q B_{\gamma,\vec{a}}^{-1} = I - P_{\gamma,\vec{a}}    
\end{align}
where we have used our convention for exponentially tilted operators defined above. Observe that 
\[
P_{\gamma,\vec{a}}^2 = \Bigl( B_{\gamma,\vec{a}} P B_{\gamma,\vec{a}}^{-1} \Bigr) \Bigl( B_{\gamma,\vec{a}} P B_{\gamma,\vec{a}}^{-1} \Bigr) =  B_{\gamma,\vec{a}} P^2 B_{\gamma,\vec{a}}^{-1} = P_{\gamma,\vec{a}}
\]
so $P_{\gamma,\vec{a}}$ is also a projection. Similarly, it can be checked that $Q_{\gamma,\vec{a}}$ is also a projection.

\section{Projected Position Operators and Exponential Localization}
\label{sec:proj-pos}
In this section, we will review our strategy for proving the existence of an exponentially localized basis of $\range{(P)}$ in two dimensions. Our overall strategy is based on our previous work \cite{2020StubbsWatsonLu1} which generalizes the proof techniques used by Nenciu-Nenciu for one dimensional systems in \cites{1998NenciuNenciu} to two dimensions. We will first review the results in one dimension in Section \ref{sec:pxp-1d} and then our extension to two dimensions in Section \ref{sec:pjypj-2d}.


\subsection{Projected Position Operators in One Dimension}
\label{sec:pxp-1d}
The idea of relating projected position operators can be traced back to a proposition by Kivelson for defining Wannier functions in non-periodic system. In \cite{1982Kivelson}, Kivelson proposed defining generalized Wannier functions to be the eigenfunctions of the projected position operator $PXP$. To support this proposal, Kivelson showed in centrosymmetric crystals the exponentially localized Wannier functions found by Kohn in \cite{1959Kohn} are in fact eigenfunctions of $PXP$. After the work by Kivelson, Niu \cite{1991Niu} argued heuristically that in one dimension the eigenfunctions of $PXP$ should decay faster than any polynomial. Later, Nenciu-Nenciu \cite{1998NenciuNenciu} showed rigorously under very general assumptions in one dimension the operator $PXP$ has discrete spectrum and its eigenfunctions are exponentially localized.

While Nenciu-Nenciu's original result was for spectral projectors of Hamiltonians of a particular form, using techniques developed in \cite{2020StubbsWatsonLu1}, one can show that in fact the eigenfunctions of $PXP$ exist and are exponentially localized only making use of the fact that $P$ admits an exponentially localized kernel. In particular, one can show the following:
\begin{theorem}
  Let $P$ be an orthogonal projector on $L^2(\R)$. If $P$ admits an exponentially localized kernel then the following is true 
  \begin{enumerate}[itemsep=1ex]
  \item The operator $PXP$ is essentially self-adjoint.
  \item The operator $PXP$ has discrete spectrum.
  \item There exists finite, positive constants $(C, \tilde{\gamma})$ such that if $\psi \in \range{(P)}$ and $PXP \psi = \eta \psi$, then 
    \[
      \int e^{2 \tilde{\gamma} \sqrt{1 + (x - \eta)^2}} |\psi(x)|^2 \, \text{\emph{d}}x \leq C.
    \]
  \end{enumerate}
  Therefore by the spectral theorem, the collection of all eigenfunctions of the operator $PXP$ form an exponentially localized, orthogonal basis for $\range{(P)}$.
\end{theorem}
Unfortunately, in two dimensions the operator $PXP$ does not generically have compact resolvent and therefore $PXP$ may not have any $L^2$-eigenvectors. 

\subsection{Projected Position Operators in Two Dimensions}
\label{sec:pjypj-2d}
To generalize the idea of using projected position operators to define Wannier functions, in recent work \cites{2020StubbsWatsonLu1} we have shown that if it is possible to decompose $P$ into a sum of ``quasi-one dimensional'' projectors, then one can extend to construct an exponentially localized basis for $\range{(P)}$ in two dimensions. The statement proven in \cites{2020StubbsWatsonLu1} is the following:
\begin{theorem}[adapted from \text{\cite[Section 7]{2020StubbsWatsonLu1}}]
  \label{thm:arma-main}
  Suppose that $P$ is an orthogonal projector on $L^2(\R^2)$ which admits an exponentially localized kernel. Suppose further that we can decompose $P$ as a sum of orthogonal projectors $\{ P_j \}_{j \in \cJ}$ where
  \[
    P = \sum_{j \in \cJ} P_j \qquad P_j P_k =
        \begin{cases}
          P_j, & j = k; \\
          0, & j \neq k.
        \end{cases}.
  \]
  If there exist constants $(C, \gamma)$ so that the following holds
  \begin{enumerate}[itemsep=1ex]
  \item Each $P_j$ admits an exponentally localized kernel with rate $\gamma$.
  \item Each $P_j$ is concentrated along a line of the form $x = \xi_j$ in the sense that for each $P_j$ there exists a $\xi_j \in \R$ such that:
    \[
      \| ( X - \xi_j ) P_{j,\gamma} \| \leq C \quad \text{and} \quad \| P_{j,\gamma} ( X - \xi_j ) \| \leq C.
    \]
  \end{enumerate}
  Then there exist constants $(C^*, \gamma^*)$ so that the following is true for each $j \in \cJ$:
  \begin{enumerate}[itemsep=1ex]
  \item The operator $P_j Y P_j$ is essentially self-adjoint.
  \item The operator $P_j Y P_j$ has discrete spectrum.
  \item If $\psi \in \range{(P_j)}$ and $P_j Y P_j \psi = \eta \psi$, then 
    \[
      \int e^{2 \gamma^* \sqrt{1 + (x - \xi_j)^2 + (y - \eta)^2}} |\psi(x,y)|^2 \, \text{\emph{d}}x \,\text{\emph{d}}y \leq C^*.
    \]
  \end{enumerate}
  Therefore by the spectral theorem, the collection of all eigenfunctions of the operators $\{ P_j Y P_j \}_{j \in \cJ}$ form an exponentially localized, orthogonal basis for $\range{(P)}$.
\end{theorem}
Given the results of Theorem~\ref{thm:arma-main}, it is natural to wonder how to construct the collection $\{ P_j \}_{j \in \cJ}$. As part of the contribution in \cites{2020StubbsWatsonLu1}, we introduced the assumption of ``uniform spectral gaps'' which allows us to explicitly construct $\{ P_j \}$ as a family of spectral projections. In particular, we say that $PXP$ has uniform spectral gaps if the following holds\footnote{As a remark, assuming that $P$ admits an exponentially localized kernel implies that $PXP$ is an essentially self-adjoint operator \cite[Lemma 4.1]{2020StubbsWatsonLu1} so its spectrum is purely real.}.
\begin{assumption}[Uniform Spectral Gaps]  \label{def:usg}
  Suppose that there exist constants $(d, D)$ such that: 
  \begin{enumerate}
  \item There exists a countable set, $\cJ$, such that:
    \[
      \sigma(PXP) = \bigcup_{j \in \cJ} \sigma_j.
    \]
  \item The distance between $\sigma_j, \sigma_k$ ($j \neq k$) is uniformly bounded from below:
    \[
      d := \min_{j \neq k} \Big( \dist( \sigma_j, \sigma_k ) \Big)  > 0.
    \]
  \item The diameter of each $\sigma_j$ is uniformly bounded:
    \[
      D := \max_{j \in \cJ} \Big( \diam( \sigma_j ) \Big) < \infty.
    \]
  \end{enumerate}
\end{assumption}
If we assume that $PXP$ has uniform spectral gaps we can define the \textit{band projectors} for $PXP$ via the spectral theorem as follows:
\begin{definition}[Band Projectors for $PXP$]
  Suppose that $PXP$ is essentially self-adjoint and satisfies the uniform spectral gaps assumption for a collection of sets $\{ \sigma_j \}_{j \in \cJ}$. We define the \textit{band projectors} for $PXP$ as the collection of orthogonal projectors $\{ P_j \}_{j \in \cJ}$ where $P_j$ is the spectral projection associated with $\sigma_j$.
\end{definition}
Intuitively speaking, since $P_j$ is a spectral projection for $PXP$ associated with the bounded set $\sigma_j$, we can expect that functions from $\range{(P_j)}$ will be concentrated in the strip $\{ (x,y) \in \R^2 : x \in \sigma_j \}$. Additionally, since there is some separation between the ``bands'' of $PXP$, using techniques from Combes-Thomas-Agmon theory, it can be shown that the projectors $P_j$ admit each admit an exponentially localized kernel with a rate which is $\Omega(d^{-1})$. Hence, by Theorem \ref{thm:arma-main}, if $PXP$ has uniform spectral gaps, then $\range{(P)}$ admits a basis of exponentially localized generalized Wannier functions.

While $PXP$ having uniform spectral gaps is a sufficient condition to prove the existence of exponentially localized generalized Wannier functions it is not a necessary condition. In \cites{2020StubbsWatsonLu2}, we show that in time reversal symmetric systems with non-zero $\Z_2$-invariant, $PXP$ cannot have uniform spectral gaps. Despite this difficulty, the main techniques for proving Theorem~\ref{thm:arma-main} do not strongly rely on the specific properties of the position operator $X$. In fact, we can replace $X$ with any essentially self-adjoint operator which is not ``too far from $X$''. More formally, we have the following result:
\begin{lemma}[adapted from \text{\cite[Lemma A.4]{2020StubbsWatsonLu1}}]
  \label{lem:arma-main}
  Suppose that $\widehat{X}$ is a symmetric operator and there exist constants $(C,C',\gamma^*)$ so that for all $\gamma \leq \gamma^*$:
  \begin{enumerate}[itemsep=1.2ex]
  \item $\| \widehat{X} - X \| \leq C$.
  \item $\| \widehat{X}_{\gamma} - \widehat{X} \| \leq C' \gamma$
  \item $P \widehat{X} P$ has uniform spectral gaps.
  \end{enumerate}
  If $\{ P_j \}_{j \in \cJ}$ are the band projectors of $P \widehat{X} P$, then the family $\{ P_j \}_{j \in \cJ}$ satisfies the assumptions of Theorem \ref{thm:arma-main}.
\end{lemma}
The main result of this paper is to show that if there exists a generalized Wannier basis for $\range{(P)}$ which is $s$-localized for $s > 5/2$, then it is possible to construct an operator $\widehat{X}$ satisfying the assumptions of Lemma \ref{lem:arma-main} and hence $\range{(P)}$ admits an exponentially localized generalized Wannier basis by Theorem \ref{thm:arma-main}. Henceforth, we will assume that we have fixed an $s$-localized basis $\{ \psi_{\alpha} \}$ with centers points $\{ \vec{\mu}_{\alpha} \}$ where $s > 5/2$.

\subsection{Paper Organization}
The remainder of this paper is organized as follows. In Section \ref{sec:bdd-density}, we show that as a consequence of the exponential localization of $P$, the center points of any generalized Wannier basis cannot be too strongly clustered; a property we term as having ``bounded density''. Using the notion of bounded density, we then construct our new position operator $\widehat{X}$ in Section \ref{sec:xhat-constr}. In the following sections, we prove the necessary properties of $\widehat{X}$. In particular, in Section \ref{sec:xhat-close} we show that $\widehat{X}$ is close to $X$ in spectral norm, in Section \ref{sec:xhat-exp-loc} we show that $\widehat{X}$ is exponentially localized, and finally in Section \ref{sec:xhat-usg} we show that $P \widehat{X} P$ has uniform spectral gaps.

\section{On the center points of generalized Wannier bases}
\label{sec:bdd-density}
A key part of our construction of $\widehat{X}$ revolves around the center points of the generalized Wannier basis $\{ \psi_{\alpha} \}$. Unlike in the periodic case, where the set of center points is closed under lattice translations, for non-periodic systems there is a large amount of freedom in the choice of center points. For sake of argument, let us fix some element of the generalized Wannier basis $\psi_{\alpha}$ with center point  $\vec{\mu}_{\alpha}$. By definition there exists a constant $C$ so that
\[
  \int_{\R^2} \la \vec{x} - \vec{\mu}_{\alpha} \ra^{2s} |\psi_{\alpha}(\vec{x})|^2 \dd{\vec{x}} \leq C.
\]
Now, let us fix some $\vec{v}_{\alpha} \in \R^2$. If we perform the mapping $\vec{\mu}_{\alpha} \mapsto \vec{\mu}_{\alpha} + \vec{v}_{\alpha}$, by triangle inequality and the fact that $|a+b|^p \leq 2^p (|a|^p + |b|^p)$ one easily checks that
\[
  \int_{\R^2} \la \vec{x} - (\vec{\mu}_{\alpha} + \vec{v}_{\alpha}) \ra^{2s} |\psi_{\alpha}(\vec{x})|^2 \dd{\vec{x}} \leq 2^{2s} (C + | \vec{v}_{\alpha} |^{2s}).
\]
Hence, so long as $| \vec{v}_{\alpha} |$ is bounded, at the price of making the constant $C$ slightly worse, we could have equally well taken the center point of $\psi_{\alpha}$ to be $\vec{\mu}_{\alpha} + \vec{v}_{\alpha}$ instead of $\vec{\mu}_\alpha$.

Despite this apparent freedom in the choice of center points, the center points of a localized generalized Wannier basis is not arbitrary; the fact that $P$ has an exponentially localized kernel puts a restriction on how densely the center points may cluster. In particular, as shown in \cite{2021LuStubbs}, we have the following result:
\begin{definition}
\label{def:bdd-density}
We say that a collection of points $\{ \vec{\mu}_{\alpha} \}_{\alpha \in \cI}$ has \textit{bounded density} if there exists a constant $M < \infty$ such that for all $\vec{x} \in \R^2$ we have
\[
    \#  \{ \alpha : | \vec{\mu}_{\alpha} - \vec{x} | \leq 1 \}  \leq M
\]  
\end{definition}
\begin{lemma}[\text{\cite[Lemma 2.1]{2021LuStubbs}}]
  \label{lem:bdd-density}
  Let $P$ be an orthogonal projector which admits an exponentially localized kernel. If $\{ \psi_{\alpha} \}_{\alpha \in \cI}$ is an $s$-localized generalized Wannier basis for $P$ for some $s > 0$, then the center points for $\{ \psi_{\alpha} \}_{\alpha \in \cI}$ have bounded density.
\end{lemma}
Since we are free to move each center point by a constant amount and all localized basis for $\range{(P)}$ have center points with bounded density, we can assume without loss of generality that these center points lie on the integer lattice:
\begin{lemma}
\label{lem:fdc}
Let $\{ \psi_{\alpha} \}_{\alpha \in \cI}$ is a $s$-localized basis with center points $\{ \vec{\mu}_{\alpha} \}_{\alpha \in \cI}$. If we additionally assume that the center points have bounded density, then we may find a positive integer $M$ so that we can relabel the basis as $\{ \psi_{\vec{m}}^{(j)} \}$ where $\vec{m} \in \Z^2$ and $j \in \{ 1, \cdots, M \}$. Furthermore, the center point of $\psi_{\vec{m}}^{(j)}$ can be taken to be $\vec{m}$ without loss of generality.
\end{lemma}
\begin{proof}[Proof Sketch]
For each $\vec{m} \in \Z^2$ let us define the unit square centered at $\vec{m}$ as follows
\[
S_{\vec{m}} := \bigg[ m_1 - \frac{1}{2}, m_1 + \frac{1}{2} \bigg) \times \bigg[ m_2 - \frac{1}{2}, m_2 + \frac{1}{2} \bigg).
\]
Since the basis $\{ \psi_{\alpha} \}_{\alpha \in \cI}$ has center points with bounded density, we know that there are at most $M$ center points contained in the square $S_{\vec{m}}$ (as it is contained in the Euclidean ball of radius 1 centered at $\vec{m}$).  Because of this, we can relabel this basis as $\{ \psi_{\vec{m}}^{(j)} \}$ where $\psi_{\vec{m}}^{(j)}$ has its center in $S_{\vec{m}}$ and $j$ is a degeneracy index which runs from $\{ 1, \cdots, M \}$. If $S_{\vec{m}}$ has fewer than $M$ center points, say it has $j^*$, then we define $\psi_{\vec{m}}^{(j)} \equiv 0$ for all $j > j^*$. Strictly speaking this enlarged set is no longer a basis, but it does not matter as we will be interested in controlling various spectral norms which are unchanged by this enlargement.
\end{proof}
By using the relabeling from Lemma \ref{lem:fdc}, we can now construct the operator $\widehat{X}$.

\section{Construction of $\widehat{X}$}
\label{sec:xhat-constr}
In the previous section, we have relabeled our $s$-localized basis as $\{ \psi_{\vec{m}}^{(j)} \}$ and discretized its center points to lie on the integer lattice. We will now use this discretization of the center points to define a new position operator $\tilde{X}$ as follows:
\begin{equation}
  \label{eq:xtilde}
  \tilde{X} := \sum_{\vec{m},j} m_1 \ket{\psi_{\vec{m}}^{(j)}} \bra{\psi_{\vec{m}}^{(j)}} + QXQ.
\end{equation}
Since $PQ = QP = 0$ it's clear that $\sigma(P \tilde{X} P) \subseteq \Z$ and therefore $P \tilde{X} P$ has uniform spectral gaps. Furthermore, it can be shown that if $\{ \psi_{\vec{m}}^{(j)} \}$ is $s$-localized with $s > 2$ then $\| \tilde{X} - X \| = O(1)$. Hence, it can be shown that $\tilde{X}$ satisfies Assumptions (1) and (3) of Lemma \ref{lem:arma-main}. Unfortunately, it doesn't seem possible to prove that $\tilde{X}$ satisfies Assumption (2) since the basis $\{ \psi_{\vec{m}}^{(j)} \}$ only decays algebraically quickly.

To address this issue, we will define a new position operator $\widehat{X}$ which modifies $\tilde{X}$ so that it satisfies Assumption (2) while still preserving Assumptions (1) and (3). Following the approach used by Hastings in \cites{2009Hastings}, we define the operator $\widehat{X}$ as follows:
\begin{equation}
  \label{eq:xhat}
  \widehat{X} := \int_{\R^2} f(t_1) f(t_2) e^{i (X t_1 + Y t_2) / \Delta } \tilde{X} e^{-i (X t_1 + Y t_2) / \Delta} \dd{t_1} \dd{t_2}
\end{equation}
Here $\Delta$ is a finite parameter to be chosen as part of our proofs and $f(t)$ is a filter function defined in terms of its Fourier transform as follows:
\[
  \hat{f}(\xi) =
  \begin{cases}
    (1 - |\xi|^2)^3 & |\xi| \leq 1 \\
    0 & |\xi| \geq 1 \\
  \end{cases}
\]
Note that since $\hat{f}(\xi)$ is $C^2(\R^2)$, $t f(t) \in L^1(\R)$. Also, note that $\int f(t) \dd{t} = \hat{f}(0) = 1$ and $\hat{f}$ is an even function so $f$ is real valued.

The effect of this modification is best understood in the special case where $\tilde{X}$ and $\widehat{X}$ are both matrices and $X$ and $Y$ are integer valued. In this case, we can find a simultaneous eigenbasis for $X$ and $Y$, $\{  e_{\vec{\lambda}} \}_{\lambda \in \Z^2}$ so that for each $\vec{\lambda} = (\lambda_1, \lambda_2) \in \Z^2$ we have that:
\[
X e_{\vec{\lambda}} = \lambda_1  e_{\vec{\lambda}} \quad  \text{and} \quad Y e_{\vec{\lambda}} = \lambda_2  e_{\vec{\lambda}}.
\]
One then easily checks that
\[
\la e_{\vec{\lambda}}, \widehat{X} e_{\vec{\mu}} \ra = \la e_{\vec{\lambda}}, \tilde{X} e_{\vec{\mu}} \ra \hat{f}\left( \frac{\lambda_1 - \mu_1}{\Delta} \right) \hat{f}\left( \frac{\lambda_2 - \mu_2}{\Delta} \right).
\]
Since $\hat{f}(0) = 1$ and $\hat{f}$ is compactly supported, this calculation shows that the entries of $\widehat{X}$ are roughly equal to the entries of $\widehat{X}$ near the diagonal and entries which are distance $\gtrsim \Delta$ away from the diagonal are set to zero. 

Now that we've defined $\widehat{X}$, to prove the main theorem it remains to show that $\widehat{X}$ satisfies the assumptions of Lemma \ref{lem:arma-main}. We prove $\widehat{X}$ that satisfies Assumption (1) in Section \ref{sec:xhat-close}, Assumption (2) in Section \ref{sec:xhat-exp-loc}, and Assumption (3) in Section \ref{sec:xhat-usg}.

\section{Closeness of $\widehat{X}$ and $X$}
\label{sec:xhat-close}
The main goal of this section is to prove the following proposition.
\begin{proposition}
  \label{prop:xhat-x-close}
  Suppose that $P$ is an orthogonal projector which admits an exponentially localized kernel (Definition~\ref{def:exp-loc-kern}). Suppose further that $P$ admits a basis which is $s$-localized for some $s > 2$, then there exists a finite constant $C > 0$ such that
  \[
    \| \widehat{X} - X \| \leq C.
  \]
\end{proposition}

We begin with a straightforward calculation. By definition of $\widehat{X}$ we have that:
\begin{align}
  \widehat{X} - X
  & = \int_{\R^2} f(t_1) f(t_2) e^{i (X t_1 + Y t_2) / \Delta} \tilde{X} e^{-i (X t_1 + Y t_2) / \Delta} \dd{t_1} \dd{t_2} - X \\
  & = \int_{\R^2} f(t_1) f(t_2)\left( e^{i (X t_1 + Y t_2) / \Delta} \tilde{X} e^{-i (X t_1 + Y t_2)/ \Delta} - X \right) \dd{t_1} \dd{t_2} \\
  & = \int_{\R^2} f(t_1) f(t_2) e^{i (X t_1 + Y t_2) / \Delta} (\tilde{X} - X) e^{-i (X t_1 + Y t_2)/ \Delta} \dd{t_1} \dd{t_2}  \label{eq:xhat-x}
\end{align}
where we have used that $\int f(t) \dd{t} = 1$ and the fact that $[X, e^{-i (X t_1 + Y t_2)}] = 0$. Therefore,
\begin{align}
  \| \widehat{X} - X \|
  & \leq \int_{\R^2} \| f(t_1) f(t_2) e^{i (X t_1 + Y t_2)/ \Delta } (\tilde{X} - X) e^{-i (X t_1 + Y t_2)/ \Delta} \| \dd{t_1} \dd{t_2} \\
  & \leq \| \tilde{X} - X \| \left( \int_{\R} | f(t_1) | \dd{t_1} \right) \left( \int_{\R} | f(t_2) | \dd{t_2} \right)
\end{align}
Since $f \in L^1(\R)$, the proposition is proved so long as we can show that $\| \tilde{X} - X\|$ is bounded. Let's recall the definition of $\tilde{X}$
\[
  \tilde{X} = \sum_{\vec{m},j} m \ket{\psi_{\vec{m}}^{(j)}} \bra{\psi_{\vec{m}}^{(j)}} + QXQ. \tag{\ref{eq:xtilde}, revisited}
\]
Now since $P + Q = I$ we have that
\begin{align}
  X - \tilde{X}
  & = (P+Q)X(P+Q) - \tilde{X} \\
  & = PXP + PXQ + QXP + QXQ - \tilde{X} \\
  & = \Big(PXP - \sum_{\vec{m},j} m \ket{\psi_{\vec{m}}^{(j)}} \bra{\psi_{\vec{m}}^{(j)}} \Big) + \Big( PXQ + QXP \Big).  \label{eq:pxp-diff}
\end{align}
Now at least formally we can write:
\begin{align}
  PXP
  & = \Bigl( \sum_{\vec{m},j} \ket{\psi_{\vec{m}}^{(j)}} \bra{\psi_{\vec{m}}^{(j)}} \Bigr) X \Bigl( \sum_{\vec{m}',j'} \ket{\psi_{\vec{m}'}^{(j')}} \bra{\psi_{\vec{m}'}^{(j')}} \Bigr) \\[1ex]
  & = \sum_{\vec{m},j} \sum_{\vec{m}',j'}  \la \psi_{\vec{m}}^{(j)}, X \psi_{\vec{m}'}^{(j')} \ra \ket{\psi_{\vec{m}}^{(j)}}  \bra{\psi_{\vec{m}'}^{(j')}}.
\end{align}
Since $\{ \psi_{\vec{m}}^{(j)} \}$ is an orthonormal basis, we have that when $(\vec{m},j) \neq (\vec{m}',j')$:
\[
  \la \psi_{\vec{m}}^{(j)}, X \psi_{\vec{m}'}^{(j')} \ra = \la \psi_{\vec{m}}^{(j)}, (X - m_1) \psi_{\vec{m}'}^{(j')} \ra.
\]
Therefore, we can express the difference from Equation~\eqref{eq:pxp-diff} as follows:
\begin{align}
  PXP - & \sum_{\vec{m},j} m \ket{\psi_{\vec{m}}^{(j)}} \bra{\psi_{\vec{m}}^{(j)}} \\
  & = \sum_{\vec{m},j} \sum_{\vec{m}',j'}  \la \psi_{\vec{m}}^{(j)}, (X - m_1) \psi_{\vec{m}'}^{(j')} \ra \ket{\psi_{\vec{m}}^{(j)}}  \bra{\psi_{\vec{m}'}^{(j')}}.
\end{align}
Hence,
\begin{align}
  \| \tilde{X} - X \| 
     & \leq  \| \sum_{\vec{m},j} \sum_{\vec{m}',j'}  \la \psi_{\vec{m}}^{(j)}, (X - m_1) \psi_{\vec{m}'}^{(j')} \ra \ket{\psi_{\vec{m}}^{(j)}}  \bra{\psi_{\vec{m}'}^{(j')}} \|  + \| PXQ \|  + \| QXP \|.
\end{align}
The term $\| PXQ \|$ can be bounded by a constant by first observing that $PXQ = P[X, Q] = -P[X,P]$ and hence $\| PXQ \| \leq \| [X, P]\|$. Since $P$ admits an exponentially localized kernel, it is easily checked that $\| [X, P] \|$ is bounded by a constant hence so is $\| PXQ \|$. The fact that $\| QXP \|$ is bounded by a constant follows similarly.

Therefore, to complete the proof of Proposition~\ref{prop:xhat-x-close}, we it suffies to show that that the following integral kernel defines a bounded operator on $L^2(\R^2)$:
\begin{equation}
\label{eq:pxp-kern}
K(\vec{x}, \vec{y}) := \sum_{\vec{m},j} \sum_{\vec{m}',j'}  \la \psi_{\vec{m}}^{(j)}, (X - m_1) \psi_{\vec{m}'}^{(j')} \ra \psi_{\vec{m}}^{(j)}(\vec{x}) \overline{\psi_{\vec{m}'}^{(j')}(\vec{y})}
\end{equation}

For these purposes, we prove something a little stronger:
\begin{proposition}
  \label{prop:abs-kern}
 If the basis $\{\psi_{\vec{m}}^{(j)} \}$ is $s$-localized with $s > 2$, then the integral operator defined by the kernel $|K(\vec{x}, \vec{y})|$ in Equation \eqref{eq:pxp-kern} is a bounded operator $L^2(\R^2) \rightarrow L^2(\R^2)$.
\end{proposition}
Note that this proposition implies that the basis expansion of $PXP$ we used above was valid.

\section{$\widehat{X}$ is exponentially localized}
\label{sec:xhat-exp-loc}
In this section, we prove the following proposition.
\begin{proposition}
  \label{prop:xhat-exp-loc}
  Suppose that $P$ is an orthogonal projector which admits an exponentially localized kernel (Definition~\ref{def:exp-loc-kern}). Suppose further that $P$ admits a basis which $s$-localized for some $s > 2$, then there exist a finite, positive constant $C$ such that for any $\gamma \geq 0$ sufficiently small:
  \begin{equation}
  \label{eq:xhat-exp-loc-eq1} 
  \| \widehat{X}_{\gamma} - \widehat{X} \| \leq C \gamma
  \end{equation}
\end{proposition}
As we saw in Equation~\eqref{eq:xhat-x} in Section~\ref{sec:xhat-close}, using the fact that $\int_{\R} f = 1$ and $[X, e^{-i (X t_1 + Y t_2) / \Delta}] = 0$ we have that:
\[
  \widehat{X} - X = \int_{\R^2} f(t_1) f(t_2) e^{i (X t_1 + Y t_2) / \Delta } (\tilde{X} - X) e^{-i (X t_1 + Y t_2) / \Delta} \dd{t_1} \dd{t_2}.
\]
Hence, since $\tilde{X} = P\tilde{X}P + QXQ$ and $X = PXP + QXQ + PXQ + QXP$ we can rewrite $\tilde{X} - X$ in the integrand above and obtain:
\begin{align}
  \widehat{X} - X
  & = \int_{\R^2} f(t_1) f(t_2) e^{i (X t_1 + Y t_2) / \Delta } (P\tilde{X}P - PXP \\
  & \hspace{14em} - QXP - PXQ) e^{-i (X t_1 + Y t_2) / \Delta } \dd{\vec{t}} \\
  & = \int_{\R^2} f(t_1) f(t_2) \Big( e^{i (X t_1 + Y t_2) / \Delta } (P\tilde{X}P - PXP) e^{-i (X t_1 + Y t_2) / \Delta} \\
  & \hspace{8em} - e^{i (X t_1 + Y t_2) / \Delta } (QXP + PXQ) e^{-i (X t_1 + Y t_2) / \Delta } \Big) \dd{t_1} \dd{t_2}.
\end{align}
To reduce clutter in the next few steps, let's define the following shorthands:
\begin{align}
  & A^{(1)} := \int_{\R^2} f(t_1) f(t_2) e^{i (X t_1 + Y t_2) / \Delta } (P\tilde{X}P - PXP) e^{-i (X t_1 + Y t_2) / \Delta} \dd{t_1} \dd{t_2} \\[1ex]
  & A^{(2)} := \int_{\R^2} f(t_1) f(t_2) e^{i (X t_1 + Y t_2) / \Delta } (Q X P + P X Q) e^{-i (X t_1 + Y t_2) / \Delta} \dd{t_1} \dd{t_2} 
\end{align}
Using this notation we clearly have that
\begin{equation}
  \label{eq:xhat-x-a1a2-1}
  \widehat{X} - X = A^{(1)} - A^{(2)}.
\end{equation}
Multiplying on the left by $B_{\gamma}$ and on the right by $B_{\gamma}^{-1}$ we have that
\begin{equation}
  \label{eq:xhat-x-a1a2-2}
  \widehat{X}_{\gamma} - X = A^{(1)}_{\gamma} - A^{(2)}_{\gamma},
\end{equation}
where we have made use of our convention for exponentially tilted operators (Section~\ref{sec:notation}). Using the identities in Equations~\eqref{eq:xhat-x-a1a2-1} and~\eqref{eq:xhat-x-a1a2-2} we can rewrite the difference we are interested in bounding as follows:
\begin{align}
  \widehat{X}_{\gamma} - \widehat{X}
  & = (\widehat{X}_{\gamma} - X) - (\widehat{X} - X) \\
  & = (A^{(1)}_{\gamma} - A^{(2)}_{\gamma}) - (A^{(1)} - A^{(2)}) \\
  & = (A^{(1)}_{\gamma} - A^{(1)}) - (A^{(2)}_{\gamma} - A^{(2)}) 
\end{align}
Hence to show that $\| \widehat{X}_{\gamma} - \widehat{X} \| \leq C \gamma$, it is enough to find constants $K_1, K_2$ so that
\begin{align}
  & \| A^{(1)}_{\gamma} - A^{(1)} \| \leq K_1 \gamma \\
  & \| A^{(2)}_{\gamma} - A^{(2)} \| \leq K_2 \gamma
\end{align}
We will show the bound for $A^{(1)}$ in Section~\ref{sec:A1-bd} and the bound for $A^{(2)}$ in Section~\ref{sec:A2-bd}.

\subsection{Bounding $\| A^{(1)}_{\gamma} - A^{(1)} \|$}
\label{sec:A1-bd}
From the calculations in Section~\ref{sec:xhat-close}, we have shown that we can write the action of $PXP - P\tilde{X}P$ in terms an integral kernel $K(\vec{x},\vec{y})$ (Equation \eqref{eq:pxp-kern}). Using this kernel, for any $g \in C_c^\infty(\R^2)$ we have that:
\[
  ((P \tilde{X} P - PXP)g)(\vec{x}) = \int_{\R^2} K(\vec{x}, \vec{y}) g(\vec{y}) \dd{\vec{y}}.
\]
We can then use this kernel to express the action of $A^{(1)}$ on any arbitrary $g \in C_c^{\infty}(\R^2)$:
\begin{align}
  (A^{(1)}g)(\vec{x})
  & = \int_{\R^2} f(t_1) f(t_2) e^{i (x_1 t_1 + x_2 t_2) / \Delta } \left( \int_{\R^2} K(\vec{x},\vec{y}) e^{-i (y_1 t_1 + y_2 t_2) / \Delta} g(\vec{y}) \right) \dd{t_1} \dd{t_2} \\
  & = \int_{\R^2} K(\vec{x},\vec{y}) g(\vec{y}) \left( \int_{\R} f(t_1) e^{i (x_1 - y_1) t_1 / \Delta} \dd{t_1}  \right) \left( \int_{\R} f(t_2) e^{i (x_2 - y_2) t_2 / \Delta} \dd{t_2}  \right) \dd{\vec{y}} \\
  & = \int_{\R^2} K(\vec{x},\vec{y}) g(\vec{y}) \hat{f}\left( \frac{x_1 - y_1}{\Delta} \right) \hat{f}\left( \frac{x_2 - y_2}{\Delta} \right) \dd{\vec{y}} 
\end{align}
Slightly abusing notation we define
\begin{equation}
  \label{eq:f-notation-abuse}
  \hat{f}\left(\frac{\vec{x} - \vec{y}}{\Delta}\right) := \hat{f}\left(\frac{x_1 - y_1}{\Delta}\right) \hat{f}\left(\frac{x_2 - y_2}{\Delta}\right).
\end{equation}
With this notation we have
\[
  (A^{(1)}g)(\vec{x}) = \int_{\R^2} K(\vec{x},\vec{y}) \hat{f}\left( \frac{\vec{x} - \vec{y}}{\Delta} \right)  g(\vec{y}) \dd{\vec{y}}.
\]
Now recall our definition for $B_{\gamma}$:
\[
B_{\gamma} = B_{\gamma,\vec{a}} = \exp\left( \gamma \sqrt{ 1 + (X-a_1)^2 + (Y-a_2)^2 } \right).
\]
Since $B_{\gamma}$ acts pointwisely, it's easy to see that
\[
  B_{\gamma} (P\tilde{X}P - PXP) B_{\gamma}^{-1}g = e^{\gamma \la \vec{x} - \vec{a} \ra} \int_{\R^2}K(\vec{x},\vec{y}) e^{-\gamma \la \vec{y} - \vec{a} \ra} g(\vec{y}) \dd{\vec{y}}.
\]
Therefore, repeating similar steps gives us that:
\begin{equation}
  \label{eq:A1_g}
  (A^{(1)}_{\gamma}g)(\vec{x}) = \int_{\R^2} K(\vec{x},\vec{y}) e^{\gamma \la \vec{x} - \vec{a} \ra}  e^{-\gamma \la \vec{y} - \vec{a} \ra} \hat{f}\left( \frac{\vec{x} - \vec{y}}{\Delta} \right) g(\vec{y}) \dd{\vec{y}}
\end{equation}
and so
\[
  ((A^{(1)}_{\gamma} - A^{(1)})g)(\vec{x}) = \int_{\R^2} K(\vec{x},\vec{y}) ( e^{\gamma \la \vec{x} - \vec{a} \ra}  e^{-\gamma \la \vec{y} - \vec{a} \ra} - 1)  \hat{f}\left( \frac{\vec{x} - \vec{y}}{\Delta} \right) g(\vec{y}) \dd{\vec{y}}.
\]

Since we are interested in the spectral norm of $A^{(1)}_{\gamma} - A^{(1)}$, we can use our expression for $(A^{(1)}_{\gamma} - A^{(1)})g$, take the inner product with any $h \in L^2(\R^2)$, and apply triangle inequality to conclude that
\begin{equation}
  \label{eq:bd1}
    \| A^{(1)}_{\gamma} - A^{(1)} \| \leq \sup_{\| g\|=\|h\|=1}
    \int_{\R^2} \int_{\R^2} \left| h(\vec{x}) K(\vec{x},\vec{y}) (e^{\gamma \la \vec{x} - \vec{a} \ra}  e^{-\gamma \la \vec{y} - \vec{a} \ra} - 1)  \hat{f}\left( \frac{\vec{x} - \vec{y}}{\Delta} \right) g(\vec{y}) \right| \dd{\vec{y}} \dd{\vec{x}} 
\end{equation}

Using reverse triangle inequality and elementary calculus one can check that
\[
  |e^{\gamma \la \vec{x} - \vec{a} \ra}  e^{-\gamma \la \vec{x} - \vec{a} \ra} - 1| \leq \gamma | \vec{x} - \vec{y} | e^{\gamma | \vec{x} - \vec{y} |}.
\]
So since $\hat{f}$ is compactly supported on $[-\Delta, \Delta]^2$ we have that:
\begin{equation}
  \label{eq:bd2}
  \| A^{(1)}_{\gamma} - A^{(1)} \| \leq \gamma (\sqrt{2} \Delta e^{\gamma \sqrt{2} \Delta}) \sup_{\| g\|=\|h\|=1} \int_{\R^2} \int_{\R^2} \left| h(\vec{x}) K(\vec{x},\vec{y}) g(\vec{y}) \right| \dd{\vec{y}} \dd{\vec{x}}.
\end{equation}
Due to Proposition~\ref{prop:abs-kern}, we know that the kernel $| K(\vec{x},\vec{y}) |$ defines a bounded operator on $L^2(\R^2)$. Hence there exists a constant $C$ so that for all $\gamma \geq 0$
\[
  \|A^{(1)}_{\gamma} - A^{(1)}\| \leq \Big(C \sqrt{2} \Delta e^{\gamma \sqrt{2} \Delta}\Big) \gamma
\]
which is what we wanted to show.

\subsection{Bounding $\|A^{(2)}_{\gamma} - A^{(2)}\|$}
\label{sec:A2-bd}
Let's begin by recalling the definitions for $A^{(2)}_{\gamma}$ and $A^{(2)}$:
\begin{align}
  & A^{(2)}_{\gamma} = \int_{\R^2} f(t_1) f(t_2) e^{i (X t_1 + Y t_2) / \Delta } (Q_{\gamma} X P_{\gamma} + P_{\gamma} X Q_{\gamma}) e^{-i (X t_1 + Y t_2) / \Delta} \dd{t_1} \dd{t_2} \\
 & A^{(2)} = \int_{\R^2} f(t_1) f(t_2) e^{i (X t_1 + Y t_2) / \Delta } (Q X P + P X Q) e^{-i (X t_1 + Y t_2) / \Delta} \dd{t_1} \dd{t_2} 
\end{align}
Hence we can write the difference we're interested in as:
\begin{align}
\int_{\R^2} f(t_1) f(t_2) e^{i (X t_1 + Y t_2) / \Delta } (Q_{\gamma} X P_{\gamma} - Q X P + P_{\gamma} X Q_{\gamma} - P X Q) e^{-i (X t_1 + Y t_2) / \Delta} \dd{t_1} \dd{t_2} 
\end{align}
We now have the following lemma which comes as a consequence of the fact that $P$ admits an exponentially localized kernel:
\begin{lemma}
\label{lem:qxp-gamma}
Suppose that $P$ is an orthogonal projector which admits an exponentially localized kernel (Definition~\ref{def:exp-loc-kern}) and $Q = I - P$. Then there exists a finite positive constants $(C, C', \gamma^*)$ such that for all $0 \leq \gamma \leq \gamma^*$
  \[
  \begin{split}
  & \| Q_{\gamma} X P_{\gamma} - QXP \| \leq C \gamma \\
  & \| P_{\gamma} X Q_{\gamma} - PXQ \| \leq C' \gamma.
  \end{split}
  \]    
\end{lemma}
\begin{proof}
  We will only prove the first bound, the second bound follows by analogous calculations. Since $Q_{\gamma} P_{\gamma} = P_{\gamma} Q_{\gamma} = 0$ we have that
  \begin{align}
       Q_{\gamma} X P_{\gamma} - QXP 
       & = Q_{\gamma} [X, P_{\gamma}] - Q [X, P] \\
       & = (Q_{\gamma} - Q + Q) [X, P_{\gamma}] - Q [X, P] \\
       & = (Q_{\gamma} - Q) [X, P_{\gamma}] - Q [X, P_{\gamma} - P] 
  \end{align}
  Using the fact that $P$ admits an exponentially localized kernel, one can easily verify that for all $\gamma$ sufficiently small there exist constants $C_1, C_2, C_3$ so that:
  \[
    \begin{split}
    & \| Q_{\gamma} - Q \| = \| P_{\gamma} - P \| \leq C_1 \gamma \\
    & \| [X, P_{\gamma}] \| \leq C_2 \\
    & \| [X, P_{\gamma} - P]  \| \leq C_3 \gamma
    \end{split}
  \]
  Hence, $\| Q_{\gamma} X P_{\gamma} - QXP \| \leq (C_1 C_2 + C_3) \gamma$ and the result is proved.
\end{proof}
Hence applying Lemma~\ref{lem:qxp-gamma} we have that 
\[
  \|A^{(2)}_{\gamma} - A^{(2)}\| \leq (C + C') \gamma \left( \int |f(t_1)| \dd{t_1} \right) \left( \int |f(t_2)| \dd{t_2} \right).
\]

\section{$P \widehat{X} P$ has uniform spectral gaps}
\label{sec:xhat-usg}
We will begin this section by first proving that $P \widehat{X} P$ is essentially self-adjoint so the notion of uniform spectral gaps makes sense. In particular, we have the following easy lemma
\begin{lemma}
  Suppose that $P$ is an orthogonal projector which admits an exponentially localized kernel (Definition~\ref{def:exp-loc-kern}). Suppose further that $P$ admits a basis which is $s$-localized for some $s > 2$. Then $P \widehat{X} P$ is essentially self-adjoint.
\end{lemma}
\begin{proof}
  Recall the definitions of $\tilde{X}$ and $\widehat{X}$:
  \begin{align}
    & \tilde{X} = \sum_{\vec{m},j} m_1 \ket{\psi_{\vec{m}}^{(j)}} \bra{\psi_{\vec{m}}^{(j)}} + QXQ \tag{\ref{eq:xtilde}, revisited}\\[.5ex]
    & \widehat{X} =  \int_{\R^2} f(t_1) f(t_2) e^{i (X t_1 + Y t_2) / \Delta } \tilde{X} e^{-i (X t_1 + Y t_2) / \Delta} \dd{t_1} \dd{t_2}. \tag{\ref{eq:xhat}, revisited}
  \end{align}
  Since $f$ is real valued and $\tilde{X}$ is a symmetric operator, it is easy to see that $\widehat{X}$ is also a symmetric operator.

  Next, notice that
  \[
    P \widehat{X} P = P X P + P (\widehat{X} - X) P
  \]
  Since $P$ admits an exponentially localized kernel the approach from \cite{1998NenciuNenciu}, it can be easily shown that in two-dimensions $PXP$ is essentially self-adjoint. Therefore, since $\| P (\widehat{X} - X) P \| \leq \| \widehat{X} - X \|$ is bounded due to our proof from Section~\ref{sec:xhat-close}, by the Kato-Rellich theorem \cite[Theorem X.12]{1975ReedSimonii}, $P \widehat{X} P$ is essentially self-adjoint.
\end{proof}
Having established essential self-adjointness, the main goal of this section is to prove the following proposition:
\begin{proposition}
  \label{prop:xhat-usg}
  Suppose that $P$ is an orthogonal projector which admits an exponentially localized kernel (Definition~\ref{def:exp-loc-kern}). Suppose further that $P$ admits a basis which is $s$-localized for some $s > 5/2$. Next, define a set of gaps $G$ as follows:
  \begin{equation}
    \label{eq:g-def}
    G = \bigcup_{m \in \Z} \biggl( m + \frac{1}{4}, m + \frac{3}{4} \biggr).
  \end{equation}
  If $\widehat{X}$ is as defined in Equation~\eqref{eq:xhat} then for $\Delta > 0$ sufficiently large, $G \subseteq \rho(P \widehat{X} P)$. Hence for such a choice of $\Delta$, $P \widehat{X} P$ has uniform spectral gaps.
\end{proposition}
The basic idea behind proving Proposition~\ref{prop:xhat-usg} is to pick some $\lambda \in G$ and consider $(\lambda - P \widehat{X} P)^{-1}$. Since by construction $\sigma(P\tilde{X}P) \subseteq \Z$ we can formally write:
\begin{align}
  (\lambda - P \widehat{X} P)^{-1}
  & = (\lambda - P \tilde{X} P + P \tilde{X} P - P \widehat{X} P)^{-1} \\
  & = (\lambda - P \tilde{X} P)^{-1} \Big(I - (P \widehat{X} P - P \tilde{X} P) (\lambda - P \tilde{X} P)^{-1} \Big)^{-1}.
\end{align}
If we can show that for some constant $C$
\begin{equation}
  \label{eq:bad-diff}
  \| (P \widehat{X} P - P \tilde{X} P) (\lambda - P \tilde{X} P)^{-1} \| \leq C \Delta^{-1},
\end{equation}
then by picking $\Delta \geq (2C)^{-1}$ we have that
\begin{align}
  \| (\lambda - P \widehat{X} P)^{-1} \|
  & \leq \| (\lambda - P \tilde{X} P)^{-1}\| \| \Big(I - (P \widehat{X} P - P \tilde{X} P) (\lambda - P \tilde{X} P)^{-1} \Big)^{-1} \| \\
  & \leq \Big( \frac{1}{4} \Big)^{-1} \Big(1 - \frac{1}{2}\Big)^{-1} = 8, 
\end{align}
where we have used that $\lambda \in G$ and $\sigma(P \tilde{X} P) \subseteq \Z$. Hence $\lambda \in \rho(P\widehat{X} P)$.

While it is possible to prove the bound in Equation~\eqref{eq:bad-diff}, we found proving this seems to require $\{ \psi_{\vec{m}}^{(j)}\}$ is $s$-localized with $s > 3$. We can slightly improve this to $s > 5/2$ by introducing decay from the resolvent $(\lambda - P \tilde{X} P)^{-1}$ ``symmetrically''.

Towards these ends, let us define the square root of $(\lambda - P \tilde{X} P)^{-1}$. Explicitly, for any $\lambda \in G$ we define $S_{\lambda}$ as follows
\begin{equation}\label{eq:defS}
  S_{\lambda} := |\lambda|^{-1/2} Q + \sum_{\vec{m},j} |\lambda - m|^{-1/2} \ket{\psi_{\vec{m}}^{(j)}} \bra{\psi_{\vec{m}}^{(j)}}
\end{equation}
Note that by construction $[S_{\lambda}, P] = 0$.

Since $P + Q = I$ and the collection $\{ \psi_{\vec{m}}^{(j)}\}$ spans $\range{(P)}$ we have that:
\begin{align}
  \lambda - P \tilde{X} P
  & = \lambda P + \lambda Q - \sum_{\vec{m},j} m \ket{\psi_{\vec{m}}^{(j)}} \bra{\psi_{\vec{m}}^{(j)}} \\
  & = \lambda Q + \sum_{\vec{m},j} (\lambda - m) \ket{\psi_{\vec{m}}^{(j)}} \bra{\psi_{\vec{m}}^{(j)}}.
\end{align}
A simple calculation shows that
\[
  S_{\lambda} (\lambda - P \tilde{X} P) S_{\lambda} = \frac{\lambda}{|\lambda|} Q + \sum_{\vec{m},j} \frac{\lambda - m}{|\lambda - m|} \ket{\psi_{\vec{m}}^{(j)}} \bra{\psi_{\vec{m}}^{(j)}}.
\]
Hence, since $\lambda \in \R$, $S_{\lambda} (\lambda - P \tilde{X} P) S_{\lambda}$ has eigenvalues $\pm 1$. 

With this definition of $S_{\lambda}$ we can now repeat similar steps to before to get
\begin{align}
  (\lambda - P \widehat{X} P)^{-1}
  & = (\lambda - P \tilde{X} P + P \tilde{X} P - P \widehat{X} P)^{-1} \\
  & = S_{\lambda} \Big(S_{\lambda} (\lambda - P \tilde{X} P) S_{\lambda} - S_{\lambda} (P \widehat{X} P - P \tilde{X} P) S_{\lambda} \Big)^{-1}  S_{\lambda}.
\end{align}
Therefore if we can show that
\begin{equation}
  \label{eq:good-diff}
  \| S_{\lambda} (P \widehat{X} P - P \tilde{X} P) S_{\lambda} \| \leq C \Delta^{-1}
\end{equation}
then by choosing $\Delta \geq (2 C)$ the previous argument implies that $\lambda \in \rho(P \widehat{X} P)$.

Let's start our proof of Equation~\eqref{eq:good-diff} by considering the difference $P \widehat{X} P - P \tilde{X} P$. Using the fact that $\int f = 1$ we have that
\begin{align}
  P \widehat{X} P& - P \tilde{X} P \\
  & = \int_{\R^2} f(t_1) f(t_2) P e^{i (X t_1 + Y t_2) / \Delta } \tilde{X} e^{-i (X t_1 + Y t_2) / \Delta} P \dd{t_1} \dd{t_2} - P \tilde{X} P \\
    & = \int_{\R^2} f(t_1) f(t_2) P \left( e^{i (X t_1 + Y t_2) / \Delta } \tilde{X} e^{-i (X t_1 + Y t_2) / \Delta} - \tilde{X} \right) P \dd{t_1} \dd{t_2}.
\end{align}
For the next few steps, let's define the difference in parenthesis as $D(t_1, t_2)$:
\[
  D(t_1,t_2) := e^{i (X t_1 + Y t_2) / \Delta } \tilde{X} e^{-i (X t_1 + Y t_2) / \Delta} - \tilde{X}.
\]
With this short-hand notation, we have that:
\begin{equation}
  \label{eq:good-diff1}
  \| S_{\lambda} (P \widehat{X} P - P \tilde{X} P) S_{\lambda} \| \leq \int_{\R^2} |f(t_1)| |f(t_2)| \| S_{\lambda} P D(t_1, t_2) P S_{\lambda} \| \dd{t_1} \dd{t_2}
\end{equation}
We will now use techniques similar to those used by Hastings \cites{2009Hastings} to control Equation~\eqref{eq:good-diff1}. One important difference between the present work and previous work is that the operators $X, Y, \tilde{X}$ are not bounded. Despite this fact, due to the multiplication on the left and right by $S_{\lambda}$, we are able to control Equation~\eqref{eq:good-diff1} and prove a similar bound to the one proved in Hastings' work \cite[Lemma 1]{2009Hastings}.

Our first step of controlling Equation~\eqref{eq:good-diff1} will be exchange the decay provided by $S_{\lambda}$ (which is diagonal in the basis $\{ \psi_{\vec{m}}^{(j)} \}$) for $\la X - \lambda \ra^{-1/2}$ (which is diagonal in the position basis). Formally, we calculate
\begin{align}
  \| S_{\lambda} & P D(t_1, t_2) P S_{\lambda} \| \\
                 & = \| S_{\lambda} P \la X - \lambda \ra^{1/2} \la X - \lambda \ra^{-1/2} D(t_1, t_2) \la X - \lambda \ra^{-1/2} \la X - \lambda \ra^{1/2} P S_{\lambda} \| \\
                 & \leq \| S_{\lambda} P \la X - \lambda \ra^{1/2} \| \| \la X - \lambda \ra^{-1/2} D(t_1, t_2) \la X - \lambda \ra^{-1/2}\| \| \la X - \lambda \ra^{1/2} P S_{\lambda} \|
\end{align}
Intuitively speaking, we should expect that $\| S_{\lambda} P \la X - \lambda \ra^{1/2}\|$ and $\| \la X - \lambda \ra^{1/2} P S_{\lambda} \|$ are both bounded since $S_{\lambda}$ is the square root of $\tilde{X}$ when restricted to $\range{(P)}$ and $X$ and $\tilde{X}$ differ by $O(1)$ in the spectral norm. Indeed, we have the following lemma:
\begin{lemma}
  \label{lem:sqrt-bd}
  Suppose that $P$ is an orthogonal projector which admits an exponentially localized kernel (Definition~\ref{def:exp-loc-kern}). Suppose further that $P$ admits an $s$-localized generalized Wannier basis for some $s > 2$. Then there exists a constant $C >0 $ such that for any $\lambda \in G$:
  \begin{align}
    & \| S_{\lambda} P \la X - \lambda \ra^{1/2} \| \leq C \\
    & \| \la X - \lambda \ra^{1/2} P S_{\lambda} \| \leq C     
  \end{align}
\end{lemma}
\begin{proof}
  Given in Appendix~\ref{sec:sqrt-bd}.
\end{proof}
Combining this lemma with the above calculation and Equation~\eqref{eq:good-diff1} we therefore conclude that
\begin{equation}
  \label{eq:good-diff2}
  \begin{split}
    \| S_{\lambda} & (P \widehat{X} P - P \tilde{X} P) S_{\lambda} \| \\
    & \leq C^2 \int_{\R^2} |f(t_1)| |f(t_2)| \| \la X - \lambda \ra^{-1/2} D(t_1, t_2) \la X - \lambda \ra^{-1/2} \| \dd{t_1} \dd{t_2}
  \end{split}
\end{equation}
For the next few steps, let us define the shorthand
\begin{equation}\label{eq:defXb}
  \tilde{X}_b := \la X - \lambda \ra^{-1/2} \tilde{X} \la X - \lambda \ra^{-1/2}.
\end{equation}
Since $\| \tilde{X} - X \| = O(1)$ it's easy to see that for a fixed value of $\lambda$, the operator $\tilde{X}_b$ is bounded as an operator acting from $L^2(\R^2) \rightarrow L^2(\R^2)$. The subscript $b$ is intended to be suggestive of the fact that $\tilde{X}_b$ is a \textit{bounded} version of $\tilde{X}$.

We can write the quantity $\la X - \lambda \ra^{-1/2} D(t_1, t_2) \la X - \lambda \ra^{-1/2}$ in terms of $\tilde{X}_b$ by commuting $\la X - \lambda \ra^{-1/2}$ with $e^{i (X t_1 + Y t_2) / \Delta}$ and $e^{-i (X t_1 + Y t_2) / \Delta}$ as follows:
\begin{align}
\la & X - \lambda \ra^{-1/2} D(t_1, t_2) \la X - \lambda \ra^{-1/2} \\[1ex]
& = \la X - \lambda \ra^{-1/2} \bigg( e^{i (X t_1 + Y t_2) / \Delta} \tilde{X} e^{-i (X t_1 + Y t_2) / \Delta} - \tilde{X} \bigg) \la X - \lambda \ra^{-1/2} \\[1ex]
    & = e^{i (X t_1 + Y t_2) / \Delta} \tilde{X}_b e^{-i (X t_1 + Y t_2) / \Delta} - \tilde{X}_b.
\end{align}
Therefore, defining $A(t_1, t_2)$ as 
\[
  A(t_1, t_2) = e^{i (X t_1 + Y t_2) / \Delta } \tilde{X}_b e^{-i (X t_1 + Y t_2) / \Delta}
\]
we see that
\[
  \la X - \lambda \ra^{-1/2} D(t_1, t_2) \la X - \lambda \ra^{-1/2} = A(t_1, t_2) - A(0,0).
\]
We now state an important proposition regarding $\tilde{X}_b$:
\begin{proposition}
  \label{prop:gaps-comm-bd}
  Suppose that $P$ is an orthogonal projector which admits an exponentially localized kernel (Definition~\ref{def:exp-loc-kern}). Suppose further that $P$ admits a basis which is $s$-localized for some $s > 5/2$. Then for any $\lambda \in G$ there exists a finite constant $C > 0$ such that
  \begin{equation}
    \begin{split}
      \| [X, \tilde{X}_b] \| & = \| \la X - \lambda \ra^{-1/2} [X, \tilde{X}] \la X - \lambda \ra^{-1/2} \| \leq C, \\
      \| [Y, \tilde{X}_b] \| & = \| \la X - \lambda \ra^{-1/2} [Y, \tilde{X}] \la X - \lambda \ra^{-1/2} \| \leq C.
    \end{split}
  \end{equation}
\end{proposition}
\begin{proof}
  Given in Appendix~\ref{sec:comm-bd-proof}.
\end{proof}
With this proposition in mind, for any $\phi \in C^\infty_c(\R^2)$ we differentiate $A(t_1, t_2) \phi$ with respect to $t_1$ to get:
\begin{align}
  \partial_{t_1} A(t_1, t_2) \phi
  & = i \Delta^{-1} e^{i (X t_1 + Y t_2) / \Delta} \bigl( X \tilde{X}_b -  \tilde{X}_b X \bigr) e^{-i (X t_1 + Y t_2) / \Delta}  \phi \\
  & = i \Delta^{-1} e^{i (X t_1 + Y t_2) / \Delta} [ X , \tilde{X}_b] e^{-i (X t_1 + Y t_2) / \Delta} \phi.
\end{align}
This differentiation step is justified for any $\phi \in C^\infty_c(\R^2)$ since $\tilde{X}_b X$ and $X \tilde{X}_b$ are both bounded operators on $C^\infty_c(\R^2)$. The fact that $\tilde{X}_b X$ is bounded is clear since $X$ is bounded on $C^\infty_c(\R^2)$ and $\tilde{X}_b$, as defined in~\eqref{eq:defXb}, is a bounded operator. The fact that $X \tilde{X}_b$ is bounded follows from the identity $X \tilde{X}_b = [X, \tilde{X}_b] + \tilde{X}_b X$ which is bounded on $C^\infty_c(\R^2)$ due to Proposition~\ref{prop:gaps-comm-bd}.

An analogous argument shows that
\[
  \partial_{t_2} A(t_1, t_2) \phi = i \Delta^{-1} e^{i (X t_1 + Y t_2) / \Delta} [ Y, \tilde{X}_b] e^{-i (X t_1 + Y t_2) / \Delta} \phi.
\]

Due to Proposition~\ref{prop:gaps-comm-bd}, it's easy to check that both $\partial_{t_1} A(t_1,t_2)\phi$ and $\partial_{t_2} A(t_1, t_2)\phi$ are continuous functions of $t_1, t_2$ so we can apply mean value theorem to conclude there exists a $(c_1, c_2) \in [0, t_1] \times [0,t_2]$ so that:
\begin{align}
  \| (A(t_1, t_2) - A(0)) \phi \| & \leq \Delta^{-1} |c_1| \| e^{i (X c_1 + Y c_2) / \Delta} [X, \tilde{X}_b] e^{-i (X c_1 + Y c_2) / \Delta} \phi \| \\
  & \hspace{2em} + \Delta^{-1} |c_2| \| e^{i (X c_1 + Y c_2) / \Delta} [Y, \tilde{X}_b ] e^{-i (X c_1 + Y c_2) / \Delta} \phi \| \\
  & \leq \Delta^{-1} \Big(|t_1| \| [X, \tilde{X}_b] \| + |t_2| \| [Y, \tilde{X}_b] \| \Big) \| \phi \| 
\end{align}
Since $C_c^\infty(\R^2)$ is dense in $L^2(\R^2)$, this implies that there exists a finite constant $C$ so that
\[
  \| \la X - \lambda \ra^{-1/2} D(t_1, t_2) \la X - \lambda \ra^{-1/2} \| \leq C \Delta^{-1} (|t_1| + |t_2|).
\]
Hence, substituting this bound into Equation~\eqref{eq:good-diff2}, we have that
\begin{align}
  \| S_{\lambda} & (P \widehat{X} P - P \tilde{X} P) S_{\lambda} \| \\
                 & \leq C \Delta^{-1}  \int_{\R^2} |f(t_1)| |f(t_2)| (|t_1| + |t_2|) \dd{t_1} \dd{t_2} \\
                 & \leq C' \Delta^{-1}, 
\end{align}
where to get the last line we have used the fact that by construction $f(t), t f(t) \in L^1(\R)$. This completes the proof of Proposition~\ref{prop:xhat-usg} and hence establishes that for $\Delta$ sufficiently large $P \widehat{X} P$ has uniform spectral gaps.

\appendix

\section{Technical Lemmas}
\label{sec:technical-lem}

We collect two technical lemmas here which will be used in other parts of the proof. 
\subsection{Decay Lemma}
\label{sec:kl-mn-bd-lem}
For the statement of our technical lemmas, for each $\vec{k} \in \Z^2$ we introduce special notation for characteristic function of the unit box centered at $\vec{k}$:
\[
\label{eq:chi-def}
  \chi_{\vec{k}}(\vec{x}) =
  \begin{cases}
    1 & \vec{x} \in \left[k_1 - \frac{1}{2}, k_1 + \frac{1}{2}\right) \times \left[k_2 - \frac{1}{2}, k_2 + \frac{1}{2}\right) \\
    0 & \text{otherwise}.
  \end{cases}
\]
Using this notation, we now state the following result:
\begin{lemma}
  \label{lem:kl-mn-bd}
  For any $s_1, s_2 \geq 0$, any $\vec{m}, \vec{k} \in \Z^2$, and any $v \in L^2(\R^2)$ 
  \begin{equation}
    \label{eq:kl-mn-bd}
    \| \chi_{\vec{k}} v \| \leq \frac{2^{s_1 + s_2} \| \chi_{\vec{k}} (|X - m_1| + 1)^{s_1} (|Y - m_2| + 1)^{s_2} v \|}{\la m_1 - k_1 \ra^{s_1} \la m_2 - k_2 \ra^{s_2}}
  \end{equation}
  where $\la x \ra$ is the Japanese bracket $\la x \ra := \sqrt{1 + |x|^2}$.
\end{lemma}
\begin{proof}
  Instead of proving Equation~\eqref{eq:kl-mn-bd} directly we will instead prove that:
  \begin{equation}
    \label{eq:kl-mn-bd1}
    \| \chi_{\vec{k}} v \| \leq \frac{\| \chi_{\vec{k}} (|X - m_1| + 1)^{s_1} (|Y - m_2| + 1)^{s_2} v \|}{(|m_1 - k_1| + 1/2)^{s_1} (|m_2 - k_2| + 1/2)^{s_2}}
  \end{equation}
  Proving Equation~\eqref{eq:kl-mn-bd1} is sufficient since for all $a \in \Z$ one can check that
  \[
    \sqrt{1 + a^2} \leq 2 ||a| + 1/2|.
  \]
  Therefore, for any $\vec{m}, \vec{k} \in \Z^2$ and any $s > 0$ we have that for $i = 1,2$:
  \begin{align}
    (|m_i - k_i| + 1/2)^{-s} & \leq 2^{s} \la m_i - k_i \ra^{-s} 
  \end{align}
  Hence, the proving Equation~\eqref{eq:kl-mn-bd1} implies Equation~\eqref{eq:kl-mn-bd}. 

  We will now prove Equation~\eqref{eq:kl-mn-bd1} in the case where $m_1 \neq k_1$ and $m_2 \neq k_2$; the other cases follow easily using similar arguments. Our main tool for proving Equation~\eqref{eq:kl-mn-bd1} will be to use ``strip'' characteristic functions in $X$ and $Y$:
  \begin{align}
    \chi\{|x_1-m_1| \leq d\}(\vec{x}) & =
    \begin{cases}
      1 & |x_1 - m_1| \leq d \\
      0 & \text{otherwise}.
    \end{cases} \\[1ex]
    \chi\{|x_2 - m_2| \leq d\}(\vec{x}) & =
    \begin{cases}
      1 & |x_2 - m_2| \leq d \\
      0 & \text{otherwise}.
    \end{cases}
  \end{align}
  The key observation is that characteristic functions
  \[
    \chi_{\vec{k}}(\vec{x}) \quad \text{and} \quad \chi\{|x_1-m_1| \leq |m_1 - k_1| - 1/2\}(\vec{x})
  \]
  have disjoint supports (up to a set of measure zero). Therefore,
  \begin{equation}
    \label{eq:cutoff-x}
    \chi_{\vec{k}}(\vec{x}) = \chi_{\vec{k}}(\vec{x})(1 - \chi\{|x_1-m_1| \leq |m_1 - k| - 1/2\}(\vec{x})).
  \end{equation}
  Using Equation~\eqref{eq:cutoff-x} for any function $v$ we have that:
  \begin{align}
    \| & \chi_{\vec{k}} v \|^2
         = \int_{\R^2} \chi_{\vec{k}} |v(\vec{x})|^2 \dd{\vec{x}} \\
       & = \int_{\R^2} \chi_{\vec{k}}(\vec{x})(1 - \chi\{|x_1-m_1| \leq |m_1 - k_1| - 1/2\}(\vec{x})) |v(\vec{x})|^2 \dd{\vec{x}} \\
       & = \int_{\R^2} \chi_{\vec{k}}(\vec{x})(1 - \chi\{|x_1-m_1| \leq |m_1 - k_1| - 1/2\})(\vec{x}) \frac{(1 + |x_1 - m_1|)^{2s_1}}{(1 + |x_1 - m_1|)^{2s_1}} |v(\vec{x})|^2 \dd{\vec{x}} \\
  \end{align}
  Since
  \[
    1 - \chi\{|x_1-m_1| \leq |m_1 - k_1| - 1/2\}(\vec{x}) = \chi\{|x_1 - m_1| > |m_1 - k_1| - 1/2\}(\vec{x})
  \]
  we have 
  \begin{multline}
    (1 - \chi\{|x_1-m_1| \leq |m_1 - k_1| - 1/2\}(\vec{x})) \frac{1}{(1 + |x_1 - m_1|)^{2s_1}} \\
     = \chi\{|x_1-m_1| > |m_1 - k_1| - 1/2\}(\vec{x}) \frac{1}{(1 + |x_1 - m_1|)^{2s_1}} \leq \frac{1}{(|m_1 - k_1| + 1/2)^{2s_1}}
  \end{multline}
  Hence
  \begin{align}
    \| \chi_{\vec{k}} v \|^2
    & \leq \frac{1}{(|m_1 - k_1| + 1/2)^{2s_1}} \int_{\R^2} \chi_{\vec{k}}(\vec{x}) (1 + |x_1 - m_1|)^{2s_1} |v(\vec{x})|^2 \dd{\vec{x}} \\[1ex]
    & = \frac{\|\chi_{\vec{k}} (1 + |X - m_1|)^{s_1} v\|^2}{(|m_1 - k_1| + 1/2)^{2s_1}}
  \end{align}
  
  We will now apply a similar argument $\|\chi_{\vec{k}} (1 + |X - m_1|)^{s_1} v\|^2$. By similar reasoning to Equation~\eqref{eq:cutoff-x}, we have that
  \begin{equation}
    \label{eq:cutoff-y}
    \chi_{\vec{k}}(\vec{x}) = \chi_{\vec{k}}(\vec{x})(1 - \chi\{|x_2- m_2| \leq |m_2 - k_2| - 1/2\}(\vec{x})).
  \end{equation}
  Therefore,
  \begin{align}
      \|\chi_{\vec{k}} & (1 + |X - m_1|)^{s_1} v\|^2 
       = \int_{\R^2} \chi_{\vec{k}}(\vec{x}) (1 + |x_1 - m_1|)^{2 s_1} |v(\vec{x})|^2 \dd{\vec{x}} \\
      & = \int_{\R^2} \chi_{\vec{k}}(\vec{x})(1 - \chi\{|x_2 - m_2| \leq |m_2 - k_2| - 1/2\}(\vec{x})) (1 + |x_1 - m_1|)^{2 s_1} |v(\vec{x})|^2 \dd{\vec{x}} \\
      & = \int_{\R^2} \chi_{\vec{k}}(\vec{x})(1 - \chi\{|x_2 - m_2| \leq |m_2 - k_2| - 1/2\}(\vec{x})) \frac{(1 + |x_1 - m_1|)^{2 s_1}(1 + |x_2 - m_2|)^{2 s_2}}{(1 + |x_2 - m_2|)^{2 s_2}} |v(\vec{x})|^2 \dd{\vec{x}}.
  \end{align}
  Hence, repeating a similar argument to before we conclude that:
  \[
    \| \chi_{\vec{k}} v \|^2 \leq \frac{\|\chi_{\vec{k}} (1 + |X - m_1|)^{s_1} (1 + |Y - m_2|)^{s_2} v\|^2}{(|m_1 - k_1| + 1/2)^{2s_1} (|m_2 - k_2| + 1/2)^{2s_2}}.
  \]
  This proves Equation~\eqref{eq:kl-mn-bd1} proving the lemma.
\end{proof}

\subsection{Product to Sum Bound}
\begin{lemma}
  \label{lem:prod-to-sum-bd}
  For any $s_1, s_2 \geq 0$, any $\vec{m} \in \R^2$, and any $v \in L^2(\R^2)$ we have the following inequality:
  \begin{align}
    \| (1 + &|X - m_1|)^{s_1} (1 + |Y - m_2|)^{s_2} v \| \\
            & \leq \| (1 + |X - m_1|)^{s_1+s_2} v \| +  \| (1 + |Y - m_2|)^{s_1 + s_2} v \|.
  \end{align}
\end{lemma}
\begin{proof}
Observe that the result is trivial if $s_1 = 0$ or $s_2 = 0$ so we can assume without loss of generality that $s_1 > 0$ and $s_2 > 0$.

  By definition we have that:
  \begin{align}
    \| (1 + & |X - m_1|)^{s_1} (1 + |Y - m_2|)^{s_2} v \|^2 \\
    &= \int_{\R^2} (1 + | x_1 - m_1 |)^{2s_1} (1 + | x_2 - m_2 |)^{2s_2} |v(\vec{x})|^2 \dd{\vec{x}}.
  \end{align}
  Since $s_1, s_2 > 0$ we can apply Young's product inequality with $p = \frac{s_1 + s_2}{s_1}$ and $q = \frac{s_1 + s_2}{s_2}$ so that
  \begin{align}
    (1 +& | x_1 - m_1 |)^{2s_1} (1 + | x_2 - m_2 |)^{2s_2} \\[1ex]
        & \leq \frac{1}{p} (1 + | x_1 - m_1 |)^{2s_1 p} + \frac{1}{q} (1 + | x_2 - m_2 |)^{2s_2 q} \\[1ex]
        & \leq \frac{1}{p} (1 + | x_1 - m_1 |)^{2(s_1 + s_2)} + \frac{1}{q} (1 + | x_2 - m_2 |)^{2(s_1 + s_2)}
  \end{align}
  Hence, using this pointwise bound:
  \begin{align}
    \| (1 + |X - m_1|)^{s_1} & (1 + |Y - m_2|)^{s_2} v \|^2 \\
                           & \leq \frac{1}{p} \| (1 + |X - m_1|)^{s_1+s_2} v \|^2 + \frac{1}{q} \| (1 + |Y - m_2|)^{s_1+s_2} v \|^2
  \end{align}
  The result follows by taking square roots, using that $\sqrt{a^2 + b^2} \leq |a| + |b|$, and observing that $\max\{ p^{-1/2}, q^{-1/2} \} \leq 1$
\end{proof}

\section{Proof of Proposition~\ref{prop:abs-kern}}
\label{sec:abs-kernel-bd}
Let us recall the integral kernel we would like to study:
\begin{equation}
K(\vec{x}, \vec{y}) := \sum_{\vec{m},j} \sum_{\vec{m}',j'}  \la \psi_{\vec{m}}^{(j)}, (X - m_1) \psi_{\vec{m}'}^{(j')} \ra \psi_{\vec{m}}^{(j)}(\vec{x}) \overline{\psi_{\vec{m}'}^{(j')}(\vec{y})}
\end{equation}
To show that this kernel defines a bounded operator from $L^2(\R^2) \rightarrow L^2(\R^2)$, we will appeal to the continuous version of Schur's test:
\begin{theorem}[Schur's Test, adapted from \text{\cite[Lemma 1.11.14]{2010Tao}}]
\label{thm:schur-tao}
Let $K : \R^2 \times \R^2 \rightarrow \C$ be a measurable function obeying the bounds
\[
\| K(\vec{x}, \cdot) \|_{L^1} \leq B_0
\]
for almost every $\vec{x} \in \R^2$, and
\[
\| K(\cdot, \vec{y}) \|_{L^1} \leq B_1
\]
for almost every $\vec{y} \in \R^2$. Then the integral operators defined by kernels $K(\vec{x}, \vec{y})$ and $|K(\vec{x}, \vec{y})|$ define a bound operator from $L^2(\R^2) \rightarrow L^2(\R^2)$ with operator norm bounded by $\sqrt{B_0 B_1}$.
\end{theorem}
Since the collection $\{ \psi_{\vec{m}}^{(j)} \}$ is pairwise orthogonal, one can easily verify that for all $\vec{m}, \vec{m}' \in \Z^2$:
\[
\la \psi_{\vec{m}}^{(j)}, (X - m_1) \psi_{\vec{m}'}^{(j')} \ra  = \la \psi_{\vec{m}}^{(j)}, (X - m_1') \psi_{\vec{m}'}^{(j')} \ra 
\]
Hence $\overline{K(\vec{x}, \vec{y})} = K(\vec{y}, \vec{x})$ so it suffices to only prove the first bound. For any $\vec{x} \in \R^2$, by triangle inequality we have that:
\[
\| K(\vec{x}, \cdot) \|_{L^1}  \leq \sum_{\vec{m},j} \sum_{\vec{m}',j'}  |\la \psi_{\vec{m}}^{(j)}, (X - m_1) \psi_{\vec{m}'}^{(j')} \ra| \, |\psi_{\vec{m}}^{(j)}(\vec{x}) | \int |\psi_{\vec{m}'}^{(j')}(\vec{y})| \dd{\vec{y}}.
\]
For the first step of our proof we appeal to the following pointwise bound proven by Marcelli, Moscolari, and Panati in \cite{2020MarcelliMoscolariPanati} 
\begin{lemma}[adapted from \text{\cite[Lemma 2.6]{2020MarcelliMoscolariPanati}}]
\label{lem:wf-pointwise}
Suppose that $P$ is an exponentially localized projector and $\{ \psi_{\vec{m}}^{(j)} \}$ is an $s$-localized generalized Wannier basis. There exists a constant $C$, depending only on $P$, so that each $\psi_{\vec{m}}^{(j)}$ satisfies the pointwise bound:
\[
|\psi_{\vec{m}}^{(j)}(\vec{x})| \leq C \la \vec{x} - \vec{m} \ra^{-s}.
\]
\end{lemma}
Due to Lemma \ref{lem:wf-pointwise}, we see that so long as $s > 2$, $\psi_{\vec{m}}^{(j)} \in L^1(\R^2)$. Hence 
\begin{align}
\| K(\vec{x}, \cdot) \|_{L^1}
&  \leq C \sum_{\vec{m},j}  \sum_{\vec{m}',j'}  |\la \psi_{\vec{m}}^{(j)}, (X - m_1) \psi_{\vec{m}'}^{(j')} \ra| \, |\psi_{\vec{m}}^{(j)}(\vec{x})| \\
 & \leq C \left( \sup_{\vec{m},j} \sum_{\vec{m}',j'}  |\la \psi_{\vec{m}}^{(j)}, (X - m_1) \psi_{\vec{m}'}^{(j')} \ra| \right) \left(\sum_{\vec{m},j}  |\psi_{\vec{m}}^{(j)}(\vec{x})| \right) \\
  & \leq C M \left( \sup_{\vec{m},j,j'} \sum_{\vec{m}'}  |\la \psi_{\vec{m}}^{(j)}, (X - m_1) \psi_{\vec{m}'}^{(j')} \ra| \right) \left(\sum_{\vec{m},j}  |\psi_{\vec{m}}^{(j)}(\vec{x})| \right)
\end{align}
where in the last line we have used the fact that $j' \in \{ 1, \cdots, M \}$. Once again using the pointwise bound from Lemma \ref{lem:wf-pointwise}, we see that the sum over $(\vec{m},j)$ is bounded by a constant for any fixed $\vec{x} \in \R^2$ so long as $s > 2$. Hence, to complete the proof of Proposition~\ref{prop:abs-kern} it suffices to show that the following sum is bounded:
\[
\label{eq:main-x-bd}
\sup_{\vec{m},j,j'} \sum_{\vec{m}'}  |\la \psi_{\vec{m}}^{(j)}, (X - m_1) \psi_{\vec{m}'}^{(j')} \ra|.
\]
For proving this bound, we will fix a choice of $\vec{m}, j, j'$ and then prove a bound which is independent of this choice. 

Our main technique for upper bounding \eqref{eq:main-x-bd} will be to insert a partition of unity of the form:
\[
  \sum_{\vec{k}} \chi_{\vec{k}}(\vec{x}) = 1
\]
where $\chi_{\vec{k}}$ is the characteristic function defined in Equation \eqref{eq:chi-def}. We prove two technical lemmas (Lemmas~\ref{lem:kl-mn-bd} and~\ref{lem:prod-to-sum-bd}) which relate the characteristic functions $\chi_{\vec{k}}$ to the basis $\{ \psi_{\vec{m}}^{(j)} \}$ in Appendix~\ref{sec:technical-lem}.

Inserting the partition of unity of characteristic functions $\sum_{\vec{k}} \chi_{\vec{k}} = 1$ we have that:
\begin{align}
  \sum_{\vec{m}'}  |\la \psi_{\vec{m}}^{(j)}, (X - m_1) \psi_{\vec{m}'}^{(j')} \ra|
                 & \leq \sum_{\vec{m}'}\sum_{\vec{k}}  |\la \chi_{\vec{k}} (X - m_1) \psi_{\vec{m}}^{(j)}, \psi_{\vec{m}'}^{(j')} \ra| \\
                 & \leq \sum_{\vec{m}'} \sum_{\vec{k}}  \| \chi_{\vec{k}} (X - m_1) \psi_{\vec{m}}^{(j)} \| \| \chi_{k\ell} \psi_{\vec{m}'}^{(j')} \| \\
                 & \leq \sum_{\vec{m}'} \sum_{\vec{k}}  \| \chi_{\vec{k}} (|X - m_1| + 1) \psi_{\vec{m}}^{(j)} \| \| \chi_{\vec{k}} \psi_{\vec{m}'}^{(j')} \|  \label{eq:l1-kern1}
\end{align}
where in the last line we have used the pointwise bound $(x - m_1)^2 \leq (|x - m_1| + 1)^2$.  Applying Lemma~\ref{lem:kl-mn-bd} with $s_1 = s_2 = 1 + \epsilon$ gives us that: 
\begin{align}
  \| \chi_{\vec{k}} &\psi_{\vec{m}'}^{(j')} \| \leq \frac{C_1 \| \chi_{\vec{k}} (|X - m_1'| + 1)^{1+\epsilon} (|Y - m_2'| + 1)^{1+\epsilon} \psi_{\vec{m}'}^{(j')} \|}{\la m_1' - k_1 \ra^{1+\epsilon} \la m'_2 - k_2 \ra^{1+\epsilon}}  \label{eq:l1-kern2} 
\end{align}
Next, applying Lemma~\ref{lem:kl-mn-bd} with $s_1 = s_2 = 1/2 + \epsilon$ gives that:
\begin{align}
  \| \chi_{\vec{k}} &(|X - m_1| + 1) \psi_{\vec{m}}^{(j)} \| \\
                  & \hspace{3em} \leq \frac{C_2 \| \chi_{\vec{k}} (|X - m_1 | + 1)^{3/2+\epsilon} (|Y - m_2| + 1)^{1/2+\epsilon} \psi_{\vec{m}}^{(j)} \|}{\la m_1 - k_1 \ra^{1/2+\epsilon} \la m_2 - k_2 \ra^{1/2+\epsilon}}.  \label{eq:l1-kern3} 
\end{align}
Applying Lemma~\ref{lem:prod-to-sum-bd} to upper bound Equation~\eqref{eq:l1-kern2} gives:
\begin{align}
  \| \chi_{k\ell} \psi_{\vec{m}'}^{(j')} \| 
  & \leq \frac{C_1 \| \chi_{\vec{k}} (|X - m_1'| + 1)^{1+\epsilon} (|Y - m_2'| + 1)^{1+\epsilon} \psi_{\vec{m}'}^{(j')} \|}{\la m'_1 - k_1 \ra^{1+\epsilon} \la m'_2 - k_2 \ra^{1+\epsilon}} \\
  & \leq \frac{C_1(\| (|X - m_1'| + 1)^{2+2\epsilon} \psi_{\vec{m}'}^{(j')} \| + \| (|Y - m_2'| + 1)^{2+2\epsilon} \psi_{\vec{m}'}^{(j')} \|)}{\la m'_1 - k_1 \ra^{1+\epsilon} \la m'_2 - k_2 \ra^{1+\epsilon}} \\
  & \leq \frac{C_3}{\la m'_1 - k_1 \ra^{1+\epsilon} \la m'_2 - k_2 \ra^{1+\epsilon}}.  \label{eq:l1-kern4}
\end{align}
In the last line, we have used the fact that by assumption $\{ \psi_{\vec{m}}^{(j)} \}$ is $s$-localized with $s > 2$, so we can pick $\epsilon$ sufficiently small so that
\[
  \| (|X - m_1'| + 1)^{2+2\epsilon} \psi_{\vec{m}'}^{(j')} \| + \| (|Y - m_2'| + 1)^{2+2\epsilon} \psi_{\vec{m}'}^{(j')} \|
\]
is bounded by a constant.

Plugging Equation~\eqref{eq:l1-kern3} and Equation~\eqref{eq:l1-kern4} into Equation~\eqref{eq:l1-kern1} we have that we can find a constant $C_4$ such that:
\begin{align}
 \sum_{\vec{m}'} & \sum_{\vec{k}}  \| \chi_{\vec{k}} (|X - m_1| + 1) \psi_{\vec{m}}^{(j)} \| \| \chi_{k\ell} \psi_{\vec{m}'}^{(j')} \|  \\
  & \leq \sum_{\vec{m}'} \sum_{\vec{k}} \frac{C_4 \| \chi_{\vec{k}} (|X - m_1| + 1)^{3/2+\epsilon} (|Y - m_2| + 1)^{1/2+\epsilon} \psi_{\vec{m}}^{(j)} \|}{\la m_1 - k_1 \ra^{1/2+\epsilon} \la m_2 - k_2 \ra^{1/2+\epsilon} \la m'_1 - k_1 \ra^{1+\epsilon} \la m'_2 - k_2 \ra^{1+\epsilon}}.
\end{align}
Treating $\vec{m}'$ as a constant, we can group the summand into two parts
\begin{align}
  & A_{\vec{k}} := \frac{\| \chi_{\vec{k}} (|X - m_1| + 1)^{3/2+\epsilon} (|Y - m_2| + 1)^{1/2+\epsilon} \psi_{\vec{m}}^{(j)} \|}{\la m_1 - k_1 \ra^{1/2+\epsilon} \la m_2 - k_2 \ra^{1/2+\epsilon}} \\
  & B_{\vec{m}'-\vec{k}} := \frac{1}{\la m'_1 - k_1 \ra^{1+\epsilon} \la m'_2 - k_2 \ra^{1+\epsilon}}.
\end{align}
With this notation the sum we wish to bound can be upper bounded as:
\begin{equation}
  \label{eq:l1-kern5}
  \sum_{\vec{m}'} \sum_{\vec{k}} A_{\vec{k}} B_{\vec{m}'-\vec{k}} \leq \Big( \sum_{\vec{k}} A_{\vec{k}}  \Big) \Big( \sup_{\vec{k}} \sum_{\vec{m}'} B_{\vec{m}'-\vec{k}} \Big) 
\end{equation}
Hence, to complete the proof it suffices to show that $A_{\vec{k}}, B_{\vec{m}'}$ are both in $\ell^1(\Z^2)$. The fact $B_{\vec{m}'} \in \ell^1(\Z^2)$ for any $\epsilon > 0$ is immediate by integral test; therefore we only need to show $A_{\vec{k}} \in \ell^1(\Z^2)$. 

Applying Cauchy-Schwarz inequality
\begin{align}
  \sum_{\vec{k}} A_{\vec{k}}
  & = \sum_{\vec{k}} \frac{\| \chi_{\vec{k}} (|X - m_1| + 1)^{3/2+\epsilon} (|Y - m_2| + 1)^{1/2+\epsilon} \psi_{\vec{m}}^{(j)} \|}{\la m_1 - k_1 \ra^{1/2+\epsilon} \la m_2 - k_2 \ra^{1/2+\epsilon}} \\
  & \leq \left( \sum_{\vec{k}}  \| \chi_{\vec{k}} (|X - m_1| + 1)^{3/2+\epsilon} (|Y - m_2| + 1)^{1/2+\epsilon} \psi_{\vec{m}}^{(j)} \|^2\right)^{1/2} \\
  & \hspace{3em} \times \left( \sum_{\vec{k}} \frac{1}{\la m_1 - k_1 \ra^{1+2\epsilon} \la m_2 - k_2 \ra^{1+2\epsilon}} \right)^{1/2} \\
  & \leq C_4 \| (|X - m_1| + 1)^{3/2+\epsilon} (|Y - m_2| + 1)^{1/2+\epsilon} \psi_{\vec{m}}^{(j)} \|
\end{align}
where in the last line we have made use of the fact that $\sum_{\vec{k}} \chi_{\vec{k}} = 1$. Since by assumption $\{ \psi_{\vec{m}}^{(j)} \}$ is $s$-localized with $s > 2$, by applying Lemma~\ref{lem:prod-to-sum-bd} we conclude that 
\[
  \| (|X - m_1| + 1)^{3/2+\epsilon} (|Y - m_2| + 1)^{1/2+\epsilon} \psi_{\vec{m}}^{(j)} \|
\]
is bounded for all $\epsilon$ sufficiently small. This completes the proof of Proposition~\ref{prop:abs-kern}

\section{Square Root Bounds}
\label{sec:sqrt-bd}
In this section, we will prove the following lemma which includes Lemma~\ref{lem:sqrt-bd} as a special case. 
\begin{lemma}
  \label{lem:sqrt-bd2}
  Suppose that $P$ is an orthogonal projector which admits an exponentially localized kernel (Definition~\ref{def:exp-loc-kern}). Suppose further that $P$ admits an $s$-localized generalized Wannier basis for some $s > 2$. 
  Then there exists a constant $C >0 $ such that for any $\lambda \in G$ (recall that $G$ is the set of gaps defined in Equation~\eqref{eq:g-def}):
  \begin{align}
    & \| S_{\lambda} P \la X - \lambda \ra^{1/2} \| \leq C \\
    & \| \la X - \lambda \ra^{1/2} P S_{\lambda} \| \leq C \\
    & \| S_{\lambda}^{-1} P \la X - \lambda \ra^{-1/2} \| \leq C \\
    & \| \la X - \lambda \ra^{-1/2} P S_{\lambda}^{-1} \| \leq C
  \end{align}
\end{lemma}
Lemma~\ref{lem:sqrt-bd2} follows as an easy corollary of the following result and the fact that $P$ admits an exponentially localized kernel:
\begin{lemma}
  \label{lem:sqrt-diff}
  Suppose that $P$ is an orthogonal projector which admits an exponentially localized kernel (Definition~\ref{def:exp-loc-kern}). Suppose further that $P$ admits an $s$-localized generalized Wannier basis for some $s > 2$. Then there exists a constant $C' >0 $ such that for any $\lambda \in G$:
  \[
    \| P S_{\lambda}^{-1} P - P \la X - \lambda \ra^{1/2} P \| \leq C'
  \]
\end{lemma}
Let's assume Lemma~\ref{lem:sqrt-diff} is true and prove Lemma~\ref{lem:sqrt-bd2}. We will return to prove Lemma~\ref{lem:sqrt-diff} in the next section (Appendix~\ref{sec:sqrt-diff}).
\begin{proof}[Proof of Lemma~\ref{lem:sqrt-bd2}]
  We will show that
  \[
  \begin{split}
    & \| S_{\lambda} P \la X - \lambda \ra^{1/2} \| \leq C \\
    & \| S_{\lambda}^{-1} P \la X - \lambda \ra^{-1/2} \| \leq C
  \end{split}
  \]
  the other two bounds follow by using the fact for any bounded operator $\| A \| = \| A^\dagger \|$. For the first bound, we calculate 
  \begin{align}
    S_{\lambda} P \la X - \lambda \ra^{1/2}
    & = S_{\lambda} P \Big( \la X - \lambda \ra^{1/2} - S_{\lambda}^{-1} + S_{\lambda}^{-1} \Big) \\[1ex]
    & = S_{\lambda} P \Big( \la X - \lambda \ra^{1/2} - S_{\lambda}^{-1} \Big) + P \\ 
    & = S_{\lambda} P \Big( \la X - \lambda \ra^{1/2} - S_{\lambda}^{-1} \Big)(P+Q) + P \\ 
    & = S_{\lambda} P \Big( \la X - \lambda \ra^{1/2} - S_{\lambda}^{-1} \Big) P + S_{\lambda} P \la X - \lambda \ra^{1/2} Q + P,
  \end{align}
  where we have used that $[P, S_{\lambda}] = 0$, $P + Q = I$, and $P Q = Q P = 0$. Therefore, we have that 
  \begin{align} 
    \| S_{\lambda} & P \la X - \lambda \ra^{1/2} \| \\
                   & \leq \| S_{\lambda} \| \Big( \| P\la X - \lambda \ra^{1/2}P - PS_{\lambda}^{-1}P \| + \| P \la X - \lambda \ra^{1/2} Q \| \Big) + 1 \\
                   & \leq \| S_{\lambda} \| \Big( \| P\la X - \lambda \ra^{1/2}P - PS_{\lambda}^{-1}P \| + \| [P, \la X - \lambda \ra^{1/2}] \| \Big) + 1.
  \end{align}
  Now observe that $\|  [P, \la X - \lambda \ra^{1/2}] \  \|$ is clearly bounded since $P$ admits an exponentially localized kernel (Definition~\ref{def:exp-loc-kern}). Hence the first bound is proved.
  
  For the second bound, using $P^2 = P$ and $[P, S_{\lambda}^{-1}] = 0$,  we calculate
  \begin{align}
      S_{\lambda}^{-1} P \la X - \lambda \ra^{-1/2}
      & = P S_{\lambda}^{-1} P \la X - \lambda \ra^{-1/2} \\
      & = P \Bigl( S_{\lambda}^{-1} - \la X - \lambda \ra^{1/2} + \la X - \lambda \ra^{1/2} \Bigr) P \la X - \lambda \ra^{-1/2} \\
      & = P \Bigl( S_{\lambda}^{-1} - \la X - \lambda \ra^{1/2} \Bigr) P \la X - \lambda \ra^{-1/2} + P \la X - \lambda \ra^{1/2} P \la X - \lambda \ra^{-1/2} \\
      & = P \Bigl( S_{\lambda}^{-1} - \la X - \lambda \ra^{1/2} \Bigr) P \la X - \lambda \ra^{-1/2} + P[\la X - \lambda \ra^{1/2}, P] \la X - \lambda \ra^{-1/2} + P 
  \end{align}
Hence, using the fact that $\| \la X - \lambda \ra^{-1/2} \| \leq 1$, we get the upper bound
\begin{align}
    \| S_{\lambda}^{-1} & P \la X - \lambda \ra^{-1/2} \| \leq \| P \Bigl( S_{\lambda}^{-1} - \la X - \lambda \ra^{1/2} \Bigr) P \| + \| [\la X - \lambda \ra^{1/2}, P] \|  + 1
\end{align}
which is bounded due to Lemma~\ref{lem:sqrt-diff} and since $P$ admits an exponentially localized kernel as before.
\end{proof}
\subsection{Proof of Lemma~\ref{lem:sqrt-diff}}
\label{sec:sqrt-diff}
The proof of this lemma follows very closely with the proof of Proposition \ref{prop:abs-kern} in 
Appendix~\ref{sec:abs-kernel-bd}. Writing out these two expressions in terms of the basis $\{ \psi_{\vec{m}}^{(j)} \}$ we have that: 
\begin{align}
  & P S_{\lambda}^{-1} P = \sum_{\vec{m},j} |m_1 - \lambda|^{1/2} \ket{\psi_{\vec{m}}^{(j)}} \bra{\psi_{\vec{m}}^{(j)}} \\
  & P \la X - \lambda \ra^{1/2} P = \sum_{\vec{m},j} \sum_{\vec{m}',j'} \la \psi_{\vec{m}}^{(j)}, \la X - \lambda \ra^{1/2} \psi_{\vec{m}'}^{(j')} \ra  \ket{\psi_{\vec{m}}^{(j)}} \bra{\psi_{\vec{m}'}^{(j')}}
\end{align}
Since $\{ \psi_{\vec{m}}^{(j)} \}$ is orthonormal, we have that when $(\vec{m},j) \neq (\vec{m}',j')$:
\begin{align}
  \la \psi_{\vec{m}}^{(j)}, \la X - \lambda \ra^{1/2} \psi_{\vec{m}'}^{(j')} \ra
  & = \la \psi_{\vec{m}}^{(j)}, \Big( \la X - \lambda \ra^{1/2} - |m_1 - \lambda|^{1/2} \Big) \psi_{\vec{m}'}^{(j')} \ra \\
  & = \la \psi_{\vec{m}}^{(j)}, \Big( \la X - \lambda \ra^{1/2} - |m_1' - \lambda|^{1/2} \Big) \psi_{\vec{m}'}^{(j')} \ra  \label{eq:sqrt-bd1}
\end{align}
Therefore, we can express the difference $P S_{\lambda}^{-1} P - P \la X - \lambda \ra^{1/2} P$ as follows: 
\begin{align}
  P & S_{\lambda}^{-1} P - P \la X - \lambda \ra^{1/2} P \\
    & = -\sum_{\vec{m},j} \sum_{\vec{m}',j'} \la \psi_{\vec{m}}^{(j)}, \Big(\la X - \lambda \ra^{1/2} - |m_1' - \lambda|^{1/2}\Big) \psi_{\vec{m}'}^{(j')} \ra  \ket{\psi_{\vec{m}}^{(j)}} \bra{\psi_{\vec{m}'}^{(j')}}
\end{align}
Following the argument from Appendix~\ref{sec:abs-kernel-bd}, by Schur's test, if $\{ \psi_{\vec{m}}^{(j)} \}$ is $s$-localized with $s > 2$ to prove Lemma \ref{lem:sqrt-diff} it suffices to show that the following quantity is bounded:
\[
\sup_{\vec{m},j,j'} \sum_{\vec{m}'} |\la \psi_{\vec{m}}^{(j)}, \Big(\la X - \lambda \ra^{1/2} - |m_1' - \lambda|^{1/2}\Big) \psi_{\vec{m}'}^{(j')} \ra|.
\]
Following the argument in Appendix~\ref{sec:abs-kernel-bd}, we fix a choice of $(\vec{m}, j, j')$ and aim to prove a bound independent of this choice. Inserting the partition of unity $\sum_{\vec{k}} \chi_{\vec{k}} = 1$ (see Equation \eqref{eq:chi-def} for the definition of $\chi_{\vec{k}}$):
\begin{align}
  \sum_{\vec{m}'} &\sum_{\vec{k}} |\la \chi_{\vec{k}} \psi_{\vec{m}}^{(j)}, \Big(\la X - \lambda \ra^{1/2} - |m_1 - \lambda|^{1/2}\Big) \psi_{\vec{m}'}^{(j')} \ra| \\
                  & \leq \sum_{\vec{m}'} \sum_{\vec{k}} \| \chi_{\vec{k}} \Big(\la X - \lambda \ra^{1/2} - |m_1 - \lambda|^{1/2}\Big) \psi_{\vec{m}}^{(j)} \| \| \chi_{\vec{k}} \psi_{\vec{m}'}^{(j')} \| \\
                  & \leq \sum_{\vec{m}'} \sum_{\vec{k}} \| \chi_{\vec{k}} (|X - m_1| + 1)^{1/2} \psi_{\vec{m}}^{(j)} \| \| \chi_{\vec{k}} \psi_{\vec{m}'}^{(j')} \|.
\end{align}
To get the last line we have used the pointwise inequality which can be easily verified for any $\lambda \in \R$:
\[
|\la x - \lambda \ra^{1/2} - |m_1 - \lambda|^{1/2}|^2 \leq |x - m_1| + 1
\]
Therefore, the quantity we want to bound is
\begin{align}
  \sum_{\vec{m}'}& \sum_{\vec{k}} \| \chi_{\vec{k}} (|X - m_1| + 1)^{1/2} \psi_{\vec{m}}^{(j)} \| \| \chi_{\vec{k}} \psi_{\vec{m}'}^{(j')} \| \label{eq:sqrt-bd2}
\end{align}
In the proof of Proposition~\ref{prop:abs-kern}, we showed that the following expression is bounded when $P$ admits a basis which is $s$-localized for $s > 2$ (see Equation~\eqref{eq:l1-kern1}):
\[
  \sup_{j'} \sum_{\vec{m}'} \sum_{\vec{k}} \| \chi_{\vec{k}} (|X - m_1| + 1) \psi_{\vec{m}}^{(j)} \| \| \chi_{\vec{k}} \psi_{\vec{m}'}^{(j')} \|.
\]
Since $(|x - m_1| + 1)^{1/2} \leq (|x - m_1| + 1)$ this calculation implies that Equation~\eqref{eq:sqrt-bd2} is bounded, completing the proof of Lemma~\ref{lem:sqrt-bd}.

\section{Proof of  Proposition~\ref{prop:gaps-comm-bd}}
\label{sec:comm-bd-proof} 
Let us start this section by recalling the proposition we want to prove:

\newtheorem*{thm:gaps-comm-bd}{Proposition~\ref{prop:gaps-comm-bd}}
\begin{thm:gaps-comm-bd}
  Suppose that $P$ is an orthogonal projector which admits an exponentially localized kernel (Definition~\ref{def:exp-loc-kern}). Suppose further that $P$ admits a basis which is $s$-localized for some $s > 5/2$. Then for any $\lambda \in G$ there exists a finite constant $C > 0$ such that
  \begin{equation}
    \label{eq:gaps-comm-bd}
    \begin{split}
      \| \la X - \lambda \ra^{-1/2} [X, \tilde{X}] \la X - \lambda \ra^{-1/2} \| \leq C, \\
      \| \la X - \lambda \ra^{-1/2} [Y, \tilde{X}] \la X - \lambda \ra^{-1/2} \| \leq C.
    \end{split}
  \end{equation}
\end{thm:gaps-comm-bd}
The core idea of this proof is to rewrite the commutators we are interested in bounding into different parts we can control.

Let's begin by considering the commutator $[Y, \tilde{X}]$. Using that $Y = PYP + PYQ + QYP + QYQ$ and $\tilde{X} = P\tilde{X}P + QXQ$ we have:
\begin{align}
  [Y, \tilde{X}]
  & = [PYP + PYQ + QYP + QYQ, P\tilde{X}P + QXQ] \\
  & = [PYP + PYQ + QYP, P\tilde{X}P] + [PYQ + QYP + QYQ, QXQ] 
\end{align}
where we have used that $PQ = QP = 0$. Grouping the terms with $PYQ + QYP$ together then gives us:
\begin{align}
  [Y, \tilde{X}] = [PYP, P\tilde{X}P] + [QYQ, QXQ] + [PYQ + QYP, \tilde{X}]
\end{align}
Performing similar calculations for $[X, \tilde{X}]$ gives us
\begin{align}
  [X, \tilde{X}]
  & = [PXP, P\tilde{X}P] + [QXQ, QXQ] + [ PXQ + QXP, \tilde{X} ] \\
  & = [PXP, P\tilde{X}P] + [ PXQ + QXP, \tilde{X} ]
\end{align}
Therefore, we have three types of terms to bound:
\begin{enumerate}[itemsep=1ex]
\item $[QYQ, QXQ]$ (see Appendix~\ref{sec:usg-bd1})
\item $[ PXQ + QXP, \tilde{X} ]$ and $[ PYQ + QYP, \tilde{X} ]$ (see Appendix~\ref{sec:usg-bd2})
\item $[PXP, P\tilde{X}P]$ and $[PYP, P\tilde{X}P]$ (see Appendix~\ref{sec:usg-bd3})
\end{enumerate}
While $[QYQ,QXQ]$ can be bounded without using the decay terms, $\la X - \lambda \ra^{-1/2}$ (see Appendix~\ref{sec:usg-bd1}), bounding the other terms requires making use of this additional decay (see Appendix~\ref{sec:usg-bd2} and Appendix~\ref{sec:usg-bd3}).

\subsection{Bounding $[QXQ,QYQ]$ term}
\label{sec:usg-bd1}
Using the fact that $Q = I - P$ and $[X,Y] = 0$ we easily calculate that
\begin{align}
  [QXQ,QYQ]
  & = QXQYQ - QYQXQ \\
  & = QX(I-P)YQ - QY(I-P)XQ \\
  & = QXYQ - QXPYQ - QYXQ + QYPXQ \\
  & = QYPXQ - QXPYQ 
\end{align}
Therefore,
\begin{align}
  \| [QXQ,QYQ] \|
  & = \| QYPXQ - QXPYQ \| \\
  & \leq \| QYP \| \| PXQ \| + \| QXP \| \| PYQ \| 
\end{align}
but this is bounded since 
\[
\begin{split}
& \| PXQ \| = \| QXP \| = \| Q [X, P] \| \leq \| [X,P] \| \\
& \| PYQ \| = \| QYP \| = \| Q [Y, P] \| \leq \| [Y,P] \| 
\end{split}
\]
and $\| [X,P] \|$ and $\| [Y,P] \|$ are clearly bounded since $P$ admits an exponentially localized kernel.

\subsection{Bounding $[ PXQ + QXP, \tilde{X} ]$ and $[ PYQ + QYP, \tilde{X} ]$ terms}
\label{sec:usg-bd2}
In this section, we will show how to bound
\[
  \la X - \lambda \ra^{-1/2} [ PYQ, \tilde{X} ] \la X - \lambda \ra^{-1/2}.
\]
Bounding $\la X - \lambda \ra^{-1/2} [ QYP, \tilde{X} ] \la X - \lambda \ra^{-1/2}$ and the other two terms follows using similar steps.

We calculate
\begin{align}
[ PYQ, \tilde{X} ]
& = - [ P[Y,P], \tilde{X} ] \\
& = - [ P [Y,P], \tilde{X} - X ] -  [ P[Y,P], X ] 
\end{align}
Hence
\begin{align}
\| & \la X - \lambda \ra^{-1/2} [ PYQ, \tilde{X} ] \la X - \lambda \ra^{-1/2} \| \\
& \leq \| [ P [Y,P], \tilde{X} - X ] \| + \| \la X - \lambda \ra^{-1/2} [ P [Y,P], X ] \la X - \lambda \ra^{-1/2} \| \\
& \leq 2 \| [Y,P] \| \| \tilde{X} - X \| + \| \la X - \lambda \ra^{-1/2} [ P [Y,P], X ] \la X - \lambda \ra^{-1/2} \|
\end{align}
The first term is bounded by a constant since $\| X - \tilde{X} \|$ is bounded (Proposition~\ref{prop:xhat-x-close}) and $P$ admits an exponentially localized kernel. Therefore, to complete the proof we only need to bound the second term.

Expanding the commutator in the second term gives
\begin{align}
\la X - \lambda \ra^{-1/2} & [ P [Y,P], X ] \la X - \lambda \ra^{-1/2} \\
& = \la X - \lambda \ra^{-1/2} \Big( P [Y, P] (X - \lambda) - (X - \lambda) P [Y, P] \Big) \la X - \lambda \ra^{-1/2}.
\end{align}
We now will show that the first term is bounded by a constant, the second term follows by a similar calculation. By inserting copies of $\la X - \lambda \ra^{1/2} \la X - \lambda \ra^{-1/2}$ we get that
\begin{align}
\la X - \lambda \ra^{-1/2} & P [Y, P] (X - \lambda) \la X - \lambda \ra^{-1/2} \\
= & \Big( \la X - \lambda \ra^{-1/2} P \la X - \lambda \ra^{1/2} \Big) \\
& \times \Big(\la X - \lambda \ra^{-1/2} [Y, P] \la X - \lambda \ra^{1/2} \Big) \\
& \times\Big( \la X - \lambda \ra^{-1/2}  (X - \lambda) \la X - \lambda \ra^{-1/2} \Big)
\end{align}
Since $P$ admits an exponentially localized kernel, it is clear that there exist constants $C, C'$ so that
\[
\begin{split}
& \| \la X - \lambda \ra^{-1/2} P \la X - \lambda \ra^{1/2}  \| \leq C \\
& \| \la X - \lambda \ra^{-1/2} [Y, P] \la X - \lambda \ra^{1/2} \| \leq C'
\end{split}
\]
Since trivially
\[
\| \la X - \lambda \ra^{-1/2}  (X - \lambda) \la X - \lambda \ra^{-1/2} \| \leq 1
\]
we conclude that $\| \la X - \lambda \ra^{-1/2} [ PYQ, \tilde{X} ] \la X - \lambda \ra^{-1/2} \|$ is bounded by a constant completing the proof.

\subsection{Bounding $[PXP, P\tilde{X}P]$ and $[PYP, P\tilde{X}P]$ terms}
\label{sec:usg-bd3}
In this section we will show how to bound the following quantities
\begin{align}
  & \la X - \lambda \ra^{-1/2} [PXP, P\tilde{X}P] \la X - \lambda \ra^{-1/2} \\
  & \la X - \lambda \ra^{-1/2} [PYP, P\tilde{X}P] \la X - \lambda \ra^{-1/2}.
\end{align}
We'll start by writing the commutators $[PXP, P\tilde{X}P]$ and $[PYP, P\tilde{X}P]$ in terms of the basis $\{ \psi_{\vec{m}}^{(j)} \}$.
\begin{align}
  [PXP, & P\tilde{X}P]
  = \left[ \sum_{\vec{m},j} \sum_{\vec{m}',j'} \la \psi_{\vec{m}}^{(j)}, X \psi_{\vec{m}'}^{(j')} \ra \ket{\psi_{\vec{m}}^{(j)}}\bra{\psi_{\vec{m}'}^{(j')}}, P\tilde{X}P \right] \\[.5ex]
        & = \sum_{\vec{m},j} \sum_{\vec{m}',j'} \la \psi_{\vec{m}}^{(j)}, X \psi_{\vec{m}'}^{(j')} \ra \left( P\tilde{X}P \ket{\psi_{\vec{m}}^{(j)}}\bra{\psi_{\vec{m}'}^{(j')}} -  \ket{\psi_{\vec{m}}^{(j)}}\bra{\psi_{\vec{m}'}^{(j')}} P\tilde{X}P \right) \\
  & = \sum_{\vec{m},j} \sum_{\vec{m}',j'}  (m_1 - m_1') \la \psi_{\vec{m}}^{(j)}, X \psi_{\vec{m}'}^{(j')} \ra \ket{\psi_{\vec{m}}^{(j)}}\bra{\psi_{\vec{m}'}^{(j')}}  \label{eq:pxp-pxtp-comm}
\end{align}
A similar calculation shows that
\begin{equation}
\label{eq:pyp-pxtp-comm}
  [PYP, P\tilde{X}P] = \sum_{\vec{m},j} \sum_{\vec{m}',j'}  (m_1 - m_1') \la \psi_{\vec{m}}^{(j)}, Y \psi_{\vec{m}'}^{(j')} \ra \ket{\psi_{\vec{m}}^{(j)}}\bra{\psi_{\vec{m}'}^{(j')}}.
\end{equation}
Since $S_{\lambda}$ is diagonal in the basis $\{ \psi_{\vec{m}}^{(j)} \}$ it will significantly simplify our arguments if we replace the decay provided by $\la X - \lambda\ra^{-1/2}$ with $S_{\lambda}$. Using Lemma~\ref{lem:sqrt-bd2} we have that
\begin{align}
  \| \la & X - \lambda \ra^{-1/2} [PXP, P\tilde{X}P] \la X - \lambda \ra^{-1/2} \| \\
         & \leq \| \la X - \lambda \ra^{-1/2} P [PXP, P\tilde{X}P] P \la X - \lambda \ra^{-1/2} \| \\
         & \leq \| \la X - \lambda \ra^{-1/2} P S_{\lambda}^{-1} \| \| S_{\lambda} [PXP, P\tilde{X}P] S_{\lambda} \| \| S_{\lambda}^{-1} P \la X - \lambda \ra^{-1/2} \| \\
         & \leq C^2 \| S_{\lambda} [PXP, P\tilde{X}P] S_{\lambda} \|
\end{align}
Similar calculations for $Y$ show that
\[
  \lVert \la X - \lambda \ra^{-1/2} [PYP, P\tilde{X}P] \la X - \lambda \ra^{-1/2} \rVert \leq C^2 \| S_{\lambda} [PYP, P\tilde{X}P] S_{\lambda} \|
\]
Hence it suffices to show that $S_{\lambda} [PXP, P\tilde{X}P] S_{\lambda}$ and $S_{\lambda} [PYP, P\tilde{X}P] S_{\lambda}$ are both bounded. Using the expressions for $[PXP, P\tilde{X}P]$ and $[PYP, P\tilde{X}P]$ from Equations~\eqref{eq:pxp-pxtp-comm} and~\eqref{eq:pyp-pxtp-comm}, we get 
\begin{align}
  S_{\lambda} & [PXP, P\tilde{X}P] S_{\lambda} \\
  & = \sum_{\vec{m},j} \sum_{\vec{m}',j'}  \frac{(m_1 - m_1')}{|\lambda - m_1|^{1/2}|\lambda - m_1'|^{1/2}} \la \psi_{\vec{m}}^{(j)}, X \psi_{\vec{m}'}^{(j')} \ra \ket{\psi_{\vec{m}}^{(j)}}\bra{\psi_{\vec{m}'}^{(j')}} \\[2ex]
  S_{\lambda} & [PYP, P\tilde{X}P] S_{\lambda} \\
  & = \sum_{\vec{m},j} \sum_{\vec{m}',j'}  \frac{(m_1 - m_1')}{|\lambda - m_1|^{1/2}|\lambda - m_1'|^{1/2}} \la \psi_{\vec{m}}^{(j)}, Y \psi_{\vec{m}'}^{(j')} \ra \ket{\psi_{\vec{m}}^{(j)}}\bra{\psi_{\vec{m}'}^{(j')}}
\end{align}
Therefore, to finish the proof of Proposition~\ref{prop:gaps-comm-bd}, we will prove the following proposition:
\begin{proposition}
  \label{prop:comm-est}
  If $\{ \psi_{\vec{m}}^{(j)} \}$ is an $s$-localized basis with $s > 5/2$, then there exists an absolute constant $C$ such that for all $\lambda \in G$ (recall that $G$ is the set of gaps defined in Equation~\eqref{eq:g-def}) we have
  \begin{align}
    & \bigl\| \sum_{\vec{m},j} \sum_{\vec{m}',j'}  \frac{(m_1 - m_1')}{|\lambda - m_1|^{1/2}|\lambda - m_1'|^{1/2}} \la \psi_{\vec{m}}^{(j)}, X \psi_{\vec{m}'}^{(j')} \ra \ket{\psi_{\vec{m}}^{(j)}}\bra{\psi_{\vec{m}'}^{(j')}} \bigr\| \leq C \label{eq:comm-bd-sum1} \\
    & \bigl\| \sum_{\vec{m},j} \sum_{\vec{m}',j'}  \frac{(m_1 - m_1')}{|\lambda - m_1|^{1/2}|\lambda - m_1'|^{1/2}} \la \psi_{\vec{m}}^{(j)}, Y \psi_{\vec{m}'}^{(j')} \ra \ket{\psi_{\vec{m}}^{(j)}}\bra{\psi_{\vec{m}'}^{(j')}} \bigr\| \leq C. \label{eq:comm-bd-sum2}
  \end{align}
\end{proposition}

\subsection{Proof of Proposition~\ref{prop:comm-est}}
\label{sec:comm-est}
In this section, we will only prove Equation~\eqref{eq:comm-bd-sum1}, Equation~\eqref{eq:comm-bd-sum2} follows by analogous steps.

To bound Equation~\eqref{eq:comm-bd-sum1}, we will use Schur's test for discrete kernels. Recall that Schur's test tells us that if $T$ is a linear operator defined by the discrete kernel $K(\vec{m},j,\vec{m}',j')$:
\[
  T f(\vec{m},j) = \sum_{\vec{m}',j'} K(\vec{m},j,\vec{m}',j') f(\vec{m}',j')
\]
and if for some real, positive functions $p, q$ we have
\begin{align}
  & \sum_{\vec{m}',j'} |K(\vec{m},j,\vec{m}',j')| q(\vec{m}',j') \leq \alpha p(\vec{m},j) \quad \text{and} \\[1ex]
  & \sum_{\vec{m}} p(\vec{m},j) |K(\vec{m},j,\vec{m}',j')| \leq \beta q(\vec{m}',j').
\end{align}
then $\| T \| \leq \sqrt{\alpha \beta}$.
Equation~\eqref{eq:comm-bd-sum1} can be viewed as the operator norm of an operator defined by the following discrete kernel:
\[
K(\vec{m},j,\vec{m}',j') := \frac{(m_1 - m_1')}{|\lambda - m_1|^{1/2}|\lambda - m_1'|^{1/2}} \la \psi_{\vec{m}}^{(j)}, X \psi_{\vec{m}'}^{(j')} \ra.
\]
Choosing $p(\vec{m},j) = |\lambda - m|^{-1/2}$ and  $q(\vec{m}',j') = |\lambda - m_1'|^{-1/2}$ and applying Schur's test we see that it is enough to find $\alpha, \beta$ such that
\begin{align}
  \sum_{\vec{m}',j'} \frac{|m_1 - m_1'|}{|\lambda - m_1'| |\lambda - m_1|^{1/2}} |\la \psi_{\vec{m}}^{(j)}, X  \psi_{\vec{m}'}^{(j')} \ra| \leq \frac{\alpha}{|\lambda - m_1|^{1/2}} \\
  \sum_{\vec{m},j} \frac{|m_1 - m_1'|}{|\lambda - m_1'|^{1/2} |\lambda - m_1|} |\la \psi_{\vec{m}}^{(j)}, X  \psi_{\vec{m}'}^{(j')} \ra| \leq \frac{\beta}{|\lambda - m_1'|^{1/2}} 
\end{align}
Multiplying both sides of the first inequality by $|\lambda - m|^{1/2}$ gives that we need to show that:
\begin{equation}
  \label{eq:comm-bd1}
  \sum_{\vec{m}',j'} \frac{|m_1 - m'|}{|\lambda - m_1'|} |\la \psi_{\vec{m}}^{(j)}, X  \psi_{\vec{m}'}^{(j')} \ra| \leq \alpha
\end{equation}
Similarly, multiplying both sides of the second inequality by $|\lambda - m_1'|^{1/2}$ gives that we need to show that:
\begin{equation}
  \label{eq:comm-bd2}
  \sum_{\vec{m},j} \frac{|m_1 - m_1'|}{|\lambda - m_1|} |\la \psi_{\vec{m}}^{(j)}, X  \psi_{\vec{m}'}^{(j') }\ra| \leq \beta
\end{equation}
Since $X$ is self-adjoint, proving the bound in Equation~\eqref{eq:comm-bd1} immediately implies Equation~\eqref{eq:comm-bd2} with $\alpha = \beta$ by performing the change of index $(\vec{m},j) \leftrightarrow (\vec{m}',j')$. Hence, we will focus on Equation~\eqref{eq:comm-bd1} for the remainder of this section.

Similar to the proof of Proposition~\ref{prop:abs-kern} in Appendix~\ref{sec:abs-kernel-bd}, our main technique for proving Equation~\eqref{eq:comm-bd1} will be inserting a partition unity of the form:
\[
  \sum_{\vec{k}} \chi_{\vec{k}}(\vec{x}) = 1
\]
where
\[
  \chi_{\vec{k}}(\vec{x}) =
  \begin{cases}
    1 & \vec{x} \in \left[k - \frac{1}{2}, k + \frac{1}{2}\right) \times \left[\ell - \frac{1}{2}, \ell + \frac{1}{2}\right) \\
    0 & \text{otherwise}.
  \end{cases}
\]

Since $\{ \psi_{\vec{m}}^{(j)} \}$ is an orthonormal basis $\la \psi_{\vec{m}}^{(j)}, \psi_{\vec{m}'}^{(j')} \ra = 0$ whenever $(\vec{m},j) \neq (\vec{m}',j')$. Therefore, we easily see that
\begin{align}
\sum_{\vec{m}',j'} & \frac{|m_1 - m_1'|}{|\lambda - m_1'|} |\la \psi_{\vec{m}}^{(j)}, X  \psi_{\vec{m}'}^{(j')} \ra|
               = \sum_{\vec{m}',j'} \frac{|m_1 - m_1'|}{|\lambda - m_1'|} |\la \psi_{\vec{m}}^{(j)}, (X - m_1)  \psi_{\vec{m}'}^{(j')} \ra|.
\end{align}
Now we can insert our partition of unity to get
\begin{align}
  \sum_{\vec{m}',j'} & \frac{|m_1 - m_1'|}{|\lambda - m_1'|} |\la \psi_{\vec{m}}^{(j)}, (X - m_1)  \psi_{\vec{m}'}^{(j')} \ra| \\
                  & \leq \sum_{\vec{m}',j'} \sum_{\vec{k}} \frac{|m_1 - m_1'|}{|\lambda - m_1'|} |\la \chi_{\vec{k}} \psi_{\vec{m}}^{(j)}, (X - m_1)  \psi_{\vec{m}'}^{(j')} \ra| \\
                  & \leq \sum_{\vec{m}',j'} \sum_{\vec{k}} \frac{|m_1 - m_1'|}{|\lambda - m_1'|} \| \chi_{\vec{k}} (X - m_1)  \psi_{\vec{m}}^{(j)} \| \| \chi_{\vec{k}}  \psi_{\vec{m}'}^{(j')} \| \\
                  & \leq M \left( \sup_{j'} \sum_{\vec{m}'} \sum_{\vec{k}} \frac{|m_1 - m_1'|}{|\lambda - m_1'|} \| \chi_{\vec{k}} (X - m_1)  \psi_{\vec{m}}^{(j)} \| \| \chi_{\vec{k}}  \psi_{\vec{m}'}^{(j')} \| \right)  \label{eq:comm-bd3}
\end{align}
Using the fact that $| m_1 - m_1' | \leq | m_1 - k_1 | + | m'_1 - k_1 |$, we can now upper bound Equation~\eqref{eq:comm-bd3} by the sum of the following two terms:
\begin{align}
  \sum_{\vec{m}'} \sum_{\vec{k}} \frac{|m_1 - k_1|}{|\lambda - m_1'|} \| \chi_{\vec{k}} (X - m_1) \psi_{\vec{m}}^{(j)} \| \| \chi_{\vec{k}} \psi_{\vec{m}'}^{(j')} \| \label{eq:comm-bd4} \\
  \sum_{\vec{m}'} \sum_{\vec{k}} \frac{|m_1' - k_1|}{|\lambda - m_1'|} \| \chi_{\vec{k}} (X - m_1) \psi_{\vec{m}}^{(j)} \| \| \chi_{\vec{k}}  \psi_{\vec{m}'}^{(j')} \| \label{eq:comm-bd5}.
\end{align}
Equations~\eqref{eq:comm-bd4} and~\eqref{eq:comm-bd5} can be bounded by essentially the same calculation. Since these two equations are asymmetric in $m_1$ and $m_1'$ it turns out that bounding Equation~\eqref{eq:comm-bd4} requires we assume the basis is $s$-localized with $s > 5/2$ whereas Equation~\eqref{eq:comm-bd5} only requires $s > 2$. Since these sums are bounded using the same techniques, we only show how to bound Equation~\eqref{eq:comm-bd4}.

For this proof, we will fix a choice of $(\vec{m},j)$ and prove a bound which is uniform in $(\vec{m},j)$. Using Lemma~\ref{lem:kl-mn-bd} with $s_1 = 1/2 + \epsilon$ and $s_2 = 1 + \epsilon$ we have that for any $\epsilon > 0$
\begin{align}
  \| \chi_{\vec{k}}  \psi_{\vec{m}'}^{(j')} \| 
                   & \leq \frac{C_1 \| \chi_{\vec{k}} (|X - m_1'| + 1)^{1/2+\epsilon} (|Y - m_2'| + 1)^{1+\epsilon} \psi_{\vec{m}'}^{(j')} \|}{\la m_1' - k_1 \ra^{1/2+\epsilon} \la m_2' - k_2 \ra^{1+\epsilon}} \\
                   & \leq \frac{C_2}{\la m_1' - k_1 \ra^{1/2+\epsilon} \la m_2' - k_2 \ra^{1+\epsilon}} 
\end{align}
where in the second line we have used Lemma~\ref{lem:prod-to-sum-bd} along with the assumption that the basis $\{ \psi_{\vec{m}}^{(j)} \}$ is $s$-localized with $s > 5/2$.

Therefore, we can upper bound Equation~\eqref{eq:comm-bd4} with
\begin{align}
  \sum_{\vec{m}'} \sum_{\vec{k}} \frac{|m_1 - k_1|}{|\lambda - m_1'|} &\| \chi_{\vec{k}} (X - m_1) \psi_{\vec{m}}^{(j)} \| \| \chi_{\vec{k}} \psi_{\vec{m}'}^{(j')} \| \\
  & \leq  C_2 \sum_{\vec{m}'} \sum_{\vec{k}} \frac{|m_1 - k_1|}{|\lambda - m_1'|} \frac{\| \chi_{\vec{k}} (X - m_1) \psi_{\vec{m}}^{(j)} \|}{\la m' - k \ra^{1/2+\epsilon} \la m'_2 - k_2 \ra^{1+\epsilon}}
\end{align}
Applying Lemma~\ref{lem:kl-mn-bd} with $s_1 = 1$ and $s_2 = 1/2+\epsilon$ we also have that 
\begin{align} 
  \| \chi_{\vec{k}} & (X - m_1) \psi_{\vec{m}}^{(j)} \| \\
                   & \leq \frac{C_3 \| \chi_{\vec{k}} (|X - m_1| + 1)^{2} (|Y - m_2| + 1)^{1/2+\epsilon} \psi_{\vec{m}}^{(j)} \|}{\la m_1 - k_1 \ra \la m_2 - k_2 \ra^{1/2+\epsilon}}  \label{eq:comm-bd4-2}
\end{align}
To reduce clutter, in the next few steps let us define:
\begin{equation}
  \label{eq:comm-bd4-adef}
  A_{\vec{k}} := \| \chi_{\vec{k}} (|X - m_1| + 1)^{2} (|Y - m_2| + 1)^{1/2+\epsilon} \psi_{\vec{m}}^{(j)} \|.
\end{equation}
Note that we have excluded the dependence on $(\vec{m},j)$ in our notation since we have fixed a choice $(\vec{m},j)$ for this proof. With this definition and the bound from Equation~\eqref{eq:comm-bd4-2} we have that
\begin{align}
  \sum_{\vec{m}'} & \sum_{\vec{k}} \frac{|m_1 - k_1|}{|\lambda - m_1'|} \frac{\| \chi_{\vec{k}} (X - m_1) \psi_{\vec{m}}^{(j)} \|}{\la m_1' - k_1 \ra^{1/2+\epsilon} \la m_2' - k_2 \ra^{1+\epsilon}} \\
               & \leq C_2 C_3 \sum_{\vec{m}'} \sum_{\vec{k}} \frac{|m_1 - k_1|}{|\lambda - m_1'|} \frac{A_{\vec{k}}}{\la m_1 - k_1 \ra \la m_2 - k_2 \ra^{1/2+\epsilon}} \frac{1}{\la m_1' - k_1 \ra^{1/2+\epsilon} \la m_2' - k_2 \ra^{1+\epsilon}} \\
               & \leq C_2 C_3 \sum_{\vec{m}'} \sum_{\vec{k}} \frac{1}{|\lambda - m_1'|} \frac{A_{\vec{k}}}{\la m_2 - k_2 \ra^{1/2+\epsilon}\la m_1' - k_1 \ra^{1/2+\epsilon} \la m_2' - k_2 \ra^{1+\epsilon}} \\
               & \leq C_2 C_3 \sum_{m_2', k_2} \frac{1}{\la m_2 - k_2\ra^{1/2+\epsilon} \la m_2' - k_2 \ra^{1+\epsilon}} \left( \sum_{m_1',k_1} \frac{A_{\vec{k}}}{|\lambda - m_1'| \la m_1' - k_1 \ra^{1/2+\epsilon}} \right)  \label{eq:comm-bd4-3}
\end{align}
Let's focus our attention on the sum over $(m_1', k_1)$
\[
  \sum_{m',k} \frac{A_{\vec{k}}}{|\lambda - m_1'| \la m_1' - k_1 \ra^{1/2+\epsilon}}
\]
Since $\lambda$, $k_2$, and $(\vec{m},j)$ are fixed this is a sum of the form
\[
  \sum_{m_1'} \sum_{k_1} a[k_1] b[m_1'] c[m_1'-k_1]
\]
which is clearly the $\ell^1$-norm of a convolution. Therefore, by Young's convolution inequality \cite[Theorem 4.2]{2001LiebLoss} with $p = 2$, $q = 1 + \frac{\epsilon}{2}$ and $r = (2 - \frac{1}{p} - \frac{1}{q})^{-1} = \frac{2(2+\epsilon)}{2+3\epsilon}$ we have that
\begin{align}
  \sum_{m_1',k_1}
  & \frac{A_{\vec{k}}}{|\lambda - m_1'| \la m_1' - k_1 \ra^{1/2+\epsilon}} \\
  & \leq \left( \sum_{k} A_{\vec{k}}^2 \right)^{1/2} \left( \sum_{m'} \frac{1}{|\lambda - m_1'|^{1 + \epsilon/2}} \right)^{1/q} \left( \sum_{m'} \frac{1}{\la m' \ra^{(1/2+\epsilon)r}} \right)^{1/r}  \label{eq:comm-bd4-4}
\end{align}
It's easy to check that 
\[
  \left(\frac{1}{2}+\epsilon\right)r = \frac{2 + 5 \epsilon + 2 \epsilon^2}{2 + 3\epsilon} =  1 + \epsilon + O(\epsilon^2)
\]
so for $\epsilon > 0$ sufficiently small the last two terms in Equation~\eqref{eq:comm-bd4-4} are bounded by a constant, $C_4$. Therefore, we conclude that
\begin{equation}
  \label{eq:comm-bd4-mk-bd}
  \sum_{m_1',k_1} \frac{A_{\vec{k}}}{|\lambda - m_1'| \la m_1' - k_1 \ra^{1/2+\epsilon}} \leq C_4 \left( \sum_{k_1} A_{\vec{k}}^2 \right)^{1/2}.
\end{equation}
Using this bound in Equation~\eqref{eq:comm-bd4-3} then gives:
\begin{align}
\sum_{\vec{m}'} & \sum_{\vec{k}} \frac{|m_1 - k_1|}{|\lambda - m_1'|} \frac{\| \chi_{\vec{k}} (X - m_1) \psi_{\vec{m}}^{(j)} \|}{\la m_1' - k_1 \ra^{1/2+\epsilon} \la m_2' - k_2 \ra^{1+\epsilon}} \\
             & \leq C_2 C_3 C_4 \sum_{m_2',k_2} \frac{1}{\la m_2 - k_2\ra^{1/2+\epsilon} \la m_2' - k_2 \ra^{1+\epsilon}} \left( \sum_{k_1} A_{\vec{k}}^2 \right)^{1/2}.
\end{align}
Similar to before, this sum is the $\ell^1$-norm of a convolution in $m_2', k_2$. Therefore, by Young's convolution inequality, with $p = q = 2$, $r = 1$ we have that
\begin{align}
  \sum_{m_2',k_2} & \frac{1}{\la m_2 - k_2\ra^{1/2+\epsilon} \la m_2' - k_2 \ra^{1+\epsilon}} \left( \sum_{k_1} A_{\vec{k}}^2 \right)^{1/2} \\
  & \leq \left( \sum_{\vec{k}} A_{\vec{k}}^2 \right)^{1/2}\left( \sum_{k_2} \frac{1}{\la m_2 - k_2\ra^{1+2\epsilon}} \right)^{1/2} \left( \sum_{m_2'} \frac{1}{ \la m_2'\ra^{1+\epsilon}} \right)
\end{align}
The last two sums are clearly bounded for any $\epsilon > 0$ so to finish the bounding Equation~\eqref{eq:comm-bd4}, we only need to show the first sum is bounded. Recalling the definition of $A_{\vec{k}}$ as Equation~\eqref{eq:comm-bd4-adef}:
\begin{align}
  \sum_{\vec{k}} A_{\vec{k}}^2
  & = \sum_{\vec{k}} \| \chi_{\vec{k}} (|X - m_1| + 1)^{2} (|Y - m_2| + 1)^{1/2+\epsilon} \psi_{\vec{m}}^{(j)} \|^2 \\
  & = \| (|X - m_1| + 1)^{2} (|Y - m_2| + 1)^{1/2+\epsilon} \psi_{\vec{m}}^{(j)} \|^2.
\end{align}
Applying Lemma~\ref{lem:prod-to-sum-bd}, we see that this quantity is bounded by a constant so long as we assume that $\{ \psi_{\vec{m}}^{(j)} \}$ is bounded with $s > 5/2$. This finishes the proof that Equation~\eqref{eq:comm-bd4} is bounded.

\bibliographystyle{plain}
\bibliography{bibliography}

\end{document}